\documentclass[11pt]{article}

\usepackage{natbib,comment,amsmath,amsthm,amssymb,mathtools,algorithm,algorithmic,url,subfigure,color}
\mathtoolsset{showonlyrefs}

\newtheorem{definition}{Definition}
\newtheorem{lemma}{Lemma}
\newtheorem{theorem}{Theorem}
\newtheorem{corollary}{Corollary}
\newtheorem{proposition}{Proposition}

\newcommand{\defword}[1]{\textbf{\boldmath{#1}}}

\begin{document}

\title{Computational Rationalization: The Inverse Equilibrium Problem}

\author{Kevin Waugh {\tt waugh@cs.cmu.edu} \\
Department of Computer Science\\
Carnegie Mellon University\\
5000 Forbes Ave\\
Pittsburgh, PA 15213, USA
\and
Brian D. Ziebart {\tt bziebart@uic.edu} \\
Department of Computer Science\\
851 S. Morgan (M/C 152)\\
Room 1120 SEO\\
Chicago, IL 60607, USA
\and
J. Andrew Bagnell {\tt dbagnell@ri.cmu.edu} \\
The Robotics Institute\\
Carnegie Mellon University\\
5000 Forbes Ave\\
Pittsburgh, PA 15213, USA}


\maketitle

\begin{abstract}
Modeling the purposeful behavior of imperfect agents from a small number of
observations is a challenging task.  When restricted to the single-agent
decision-theoretic setting, inverse optimal control techniques assume that
observed behavior is an approximately optimal solution to an unknown decision
problem.  These techniques learn a utility function that explains the example
behavior and can then be used to accurately predict or imitate future behavior
in similar observed or unobserved situations.

In this work, we consider similar tasks in competitive and cooperative
multi-agent domains.  Here, unlike single-agent settings, a player cannot
myopically maximize its reward; it must speculate on how the other agents may
act to influence the game's outcome.  Employing the game-theoretic notion of
regret and the principle of maximum entropy, we introduce a technique for
predicting and generalizing behavior.
\end{abstract}


\section{Introduction}

Predicting the actions of others in complex and strategic settings is an
important facet of intelligence that guides our interactions---from walking in
crowds to negotiating multi-party deals.  Recovering such behavior from merely
a few observations is an important and challenging machine learning task.

While mature computational frameworks for decision-making have been developed
to {\bf prescribe} the behavior that an agent {\em should} perform, such
frameworks are often ill-suited for {\bf predicting} the behavior that an agent
{\em will} perform.  Foremost, the standard assumption of decision-making
frameworks that a criteria for preferring actions ({\em e.g.}, costs,
motivations and goals) is known {\it a priori} often does not hold.  Moreover,
real behavior is typically not consistently optimal or completely rational; it
may be influenced by factors that are difficult to model or subject to various
types of error when executed. Meanwhile, the standard tools of statistical
machine learning ({\em e.g.}, classification and regression) may be equally
poorly matched to modeling purposeful behavior; an agent's goals often
succinctly, but implicitly, encode a strategy that would require a tremendous
number of observations to learn.

A natural approach to mitigate the complexity of recovering a full strategy for
an agent is to consider identifying a compactly expressed utility function that
{\em rationalizes} observed behavior: that is, identify rewards for which the
demonstrated behavior is optimal and then leverage these rewards for future
prediction. Unfortunately, the problem is fundamentally ill-posed: in general,
many reward functions can make behavior seem rational, and in fact, the
trivial, everywhere zero reward function makes {\bf all} behavior appear
rational~\citep{ng2000algorithms}. Further, after removing such trivial reward
functions, there may be {\bf no} reward function for which the demonstrated
behavior is optimal as agents may be imperfect or the world they operate in may
only be approximately represented.

In the single-agent decision-theoretic setting, inverse optimal control methods
have been used to bridge this gap between the prescriptive frameworks and
predictive applications~\citep{abbeel2004,ratliff2006,ziebart2008,ziebart2010}.
Successful applications include learning and prediction tasks in personalized
vehicle route planning~\citep{ziebart2008}, predictive cursor
control~\citep{ziebart2012}, robotic crowd navigation~\citep{henry2010},
quadruped foot placement and grasp selection~\citep{ratliff2009}.  A reward
function is learned by these techniques that both explains demonstrated
behavior and approximates the optimality criteria of decision-theoretic
frameworks.

As these methods only capture a single reward function and do not reason about
competitive or cooperative motives, inverse optimal control proves inadequate
for modeling the strategic interactions of multiple agents.  In this article,
we consider the game-theoretic concept of regret as a stand-in for the
optimality criteria of the single-agent work.  As with the inverse optimal
control problem, the result is fundamentally ill-posed. We address this by
requiring that for any utility function linear in known features, our learned
model must have no more regret than that of the observed behavior.  We
demonstrate that this requirement can be re-cast as a set of equivalent convex
constraints that we denote the {\em inverse correlated equilibrium} (ICE)
polytope.

As we are interested in the effective prediction of behavior, we will use a
maximum entropy criteria to select behavior from this polytope.  We demonstrate
that optimizing this criteria leads to mini-max optimal prediction of behavior
subject to approximate rationality.  We consider the dual of this problem and
note that it generalizes the traditional log-linear maximum entropy family of
problems~\citep{della2002inducing}.  We provide a simple and computationally
efficient gradient-based optimization strategy for this family and show that
only a small number of observations are required for accurate prediction and
transfer of behavior.  We conclude by considering a variety of experimental
results, ranging from predicting travel routes in a synthetic routing game to a
market-entry econometric data-analysis exploring regulatory effects on hotel
chains in Texas.

Before we formalize imitation learning in matrix games, motivate our
assumptions and describe and analyze our approach, we review related work.

\section{Related Work}

Many research communities are interested in computational modeling of human
behavior and, in particular, in modeling rational and strategic behavior with
incomplete knowledge of utility.  Here we contrast the contributions of three
communities by overviewing their interests and approaches.  We conclude by
describing our contribution in the same light.

The econometrics community combines microeconomics and statistics to
investigate the empirical properties of markets from sales records, census data
and other publicly available statistics.  McFadden first considered estimating
consumer preferences for transportation by assuming them to be rational utility
maximizers~\citep{mcfadden74}.  Berry, Levinsohn and Pakes estimate both supply
and demand-side preferences in settings where the firms must price their goods
strategically~\citep{berry95}.   Their initial work described a procedure for
measuring the desirability of certain automobile criteria, such as fuel economy
and features like air conditioning, to determine substitution effects.

The Berry, Levinsohn and Pakes approach and its derivatives can be crudely
described as model-fitting techniques.  First, a parameterized
class of utility functions are assumed for both the producers and consumers.
Variables that are unobservable to the econometrician, such as internal
production costs and certain aspects of the consumer's preferences, are known
as {\em shocks} and are modeled as independent random variables.  The draws of
these random variables are known to the market's participants, but only their
distributions are known to the econometrician.  Second, an equilibrium pricing
model is assumed for the producers.  The consumers are typically assumed to be
utility maximizers having no strategic interactions with in the market.
Finally, an estimation technique is optimistically employed to determine a set
of parameter values that are consistent with the observed behavior. Ultimately,
it is from these parameter values that one derives insight into the
unobservable characteristics of the market.  Unfortunately, neither efficient
sample nor computational complexity bounds are generally available using this
family of approaches. 

A variety of questions have been investigated by econometricians using this
line of reasoning.  Petrin investigated the competitive advantage of being the
first producer in a market by considering the introduction of the minivan to
the North American automotive industry~\citep{petrin2002}.  Nevo provided
evidence against price-fixing in the breakfast cereal market by measuring the
effects of advertising~\citep{nevo01}.  Others have examined the mid-scale hotel
market to determine the effects of different regulatory
practices~\citep{suzuki2010} and how overcapacity can be used to deter
competition~\citep{conlin06}.  As a general theme, the econometricians are
interested in the intentions that guide behavior.  That is, the observed
behavior is considered to be the truth and the decision-making framework used
by the producers and consumers is known {\em a priori}.

The decision theory community is interested in human behavior on a more
individual level.  They, too, note that out-of-the-box game theory fails to
explain how people act in many scenarios.  As opposed to viewing this as a flaw
in the theories, they focus on both how to alter the games that model our
interactions in addition to devising human-like decision-making algorithms.
The former can be achieved through modifications to the players' utility
functions, which are known {\em a priori}, to incorporate notions such as risk
aversion and spite~\citep{myers60,erev2008}.  latter approaches often tweak
learning algorithms by integrating memory limitations or emphasizing recent or
surprising observations~\citep{camerer1999,erev2005}.

the Iterative Weighting and Sampling algorithm (I-SAW) is more likely to choose
the action with the highest estimated utility, but recent observations are
weighted more highly and, in the absence of a surprising observation, the
algorithm favors repeating previous actions~\citep{erev2010}.  Memory
limitations, or more generally bounded rationality, have also led to novel
equilibrium concepts such as the quantal response
equilibrium~\citep{mckelvey1995}.  This concept assumes the players' strategies
have faults, but that small errors, in terms of forgone utility, are much more
common than large errors.  Contrasting with the econometricians, the decision
theory community is mainly interested in the algorithmic process of human
decision-making.  The players' preferences are known and observed behavior
serves only to validate or inform an experimental hypothesis.

Finally, the machine learning community is interested in predicting and
imitating the behavior of humans and expert systems.  Much work in this area
focuses on the single-agent setting and in such cases it is known as
\emph{inverse optimal control} or \emph{inverse reinforcement
learning}~\citep{abbeel2004,ng2000algorithms}.  Here, the observed behavior is
assumed to be an approximately optimal solution to an unknown decision problem.
At a high level, known solutions typically summarize the behavior as parameters
to a low dimensional utility function.  A number of methods have been
introduced to learn these weights, including margin-based
methods~\citep{ratliff2006} that can utilize black box optimal control or
planning software, as well as maximum entropy-based methods with desirable
predictive guarantees~\citep{ziebart2008}.  These utility weights are then used
to mimic the behavior in similar situations through a decision-making
algorithm.  Unlike the other two communities, it is the predictive performance
of the learned model that is most pivotal and noisy observations are expected
and managed by those techniques.

This article extends our prior publication---a novel maximum entropy approach
for predicting behavior in strategic multi-agent
domains~\citep{waugh11,waugh11arXiv}. We focus on the computationally and
statistically efficient recovery of good estimates of behavior (the only
observable quantity) by leveraging rationality assumptions.  The work presented
here further develops those ideas in two key ways. First, we consider
distributions over games and parameterized families of deviations using the
notion of conditional entropy.  Second, this work enables more fine-grained
assumptions regarding the players' possible preferences. Finally, this work
presents the analysis of data-sets from both the econometric and decision
theory communities, comparing and contrasting the methods presented with
statistical methods that are blind to the strategic aspects of the domains.

Before describing our approach, we will introduce the necessary notation and
background.

\section{Preliminaries}

\subsection*{Notation}

\newcommand{\Vsym}{V} 
\newcommand{\Vdualsym}{\Vsym}
\newcommand{\Ksym}{\mathcal{K}} 
\newcommand{\Kdualsym}{\mathcal{K}^*}
\newcommand{\Rsym}{\mathbb R}
\newcommand{\RKsym}{\Rsym^{K}}
\newcommand{\inner}[2]{\langle #1,#2 \rangle}
\newcommand{\func}[3]{#1 : #2\rightarrow #3}
\newcommand{\argmax}{\operatornamewithlimits{argmax}}
\newcommand{\norm}[1]{\Vert#1\Vert}
\newcommand{\abs}[1]{\vert#1\vert}

Let $\Vsym$ be a Hilbert space with an inner product
$\func{\inner{\cdot}{\cdot}}{\Vdualsym\times\Vsym}{\Rsym}$.  For any set
$\Ksym\subseteq\Vsym$, let $\Kdualsym=\{x \mid \inner{x}{y} \ge 0, \forall
y\in\Ksym \}$ be its dual cone.  We let $\norm{v}_2 = \sqrt{\inner{v}{v}}$,
and, if $V$ is of finite dimension with orthonormal basis $\{e_1,\ldots,e_K\}$,
let $\norm{v}_1 = \sum_{k=1}^K \abs{\alpha_k}$ where $v = \sum_{k=1}^K\alpha_k
e_k$. Typically, we will take $\Vsym=\RKsym$ and use the standard inner
product.

\subsection*{Game Theory} 

Matrix games are the canonical tool of game theorists for representing
strategic interactions ranging from illustrative toy problems, such as the
``Prisoner's Dilemma" and the ``Battle of the Sexes" games, to important
negotiations, collaborations, and auctions.  Unlike the traditional
definition~\citep{osborne1994course}, in this work we model games
where only the features of the players' utility are known and not 
the utilities themselves.
\newcommand{\Gamesym}{\Gamma}
\newcommand{\players}{\mathcal N}
\newcommand{\outcomes}{\mathcal A}
\newcommand{\playeri}{i}
\newcommand{\utilityi}[1]{u^{#1}_{\playeri}}
\newcommand{\numplayers}{N}
\newcommand{\numactions}{A}
\newcommand{\actionsi}{A_{\playeri}}
\begin{definition} 
A \defword{vector-valued normal-form game} is a tuple
$\Gamesym=(\players,\outcomes,\utilityi{\Gamesym})$ where
\begin{itemize}
\item $\players$ is the finite set of the game's $\numplayers$ \defword{players}, 
\item $\outcomes=\times_{\playeri\in\players}\actionsi$ is the set of the
	game's \defword{outcomes} or \defword{joint-actions}, where 
\item $\actionsi$ is the finite set of \defword{actions} or
	\defword{strategies} for player $\playeri$, and 
\item $\func{\utilityi{\Gamesym}}{\outcomes}{\Vdualsym}$ is the
	\defword{utility feature function} for player $\playeri$.  
\end{itemize}
We let $\numactions=\max_{\playeri\in\players}\abs{\actionsi}$.
\end{definition}

\newcommand{\outcome}{a} 
\newcommand{\wsym}{w}
\newcommand{\utilityif}[2]{\utilityi{#1}(#2)}
\newcommand{\utilityiw}[3]{\utilityi{#1}(#2|#3)}
\newcommand{\wstarsym}{w^*}
Players aim to maximize their \defword{utility}, a quantity measuring happiness
or individual well-being.  We assume that the players' utility is a common
linear function of the utility features.  This will allow us to treat the
players anonymously should we so desire.  One can expand the utility feature
space if separate utility functions are desired.  We write the utility for
player $\playeri$ at outcome $\outcome$ under utility function $\wsym\in\Vsym$
as
\begin{align}
\utilityiw{\Gamesym}{\outcome}{\wsym} &=\inner{\utilityif{\Gamesym}{\outcome}}{\wsym}.
\end{align}

In contrast to the standard definition of normal-form games, where the utility
functions for game outcomes are known, in this work we assume that the
\defword{true utility function}, formed by $\wstarsym$, which governs observed
behavior, is unknown.  This allows us to model real-world scenarios where a
cardinal utility is not available or is subject to personal taste.  Consider,
for instance, a scenario where multiple drivers each choose a route among
shared roads.  Each outcome, which specifies a travel plan for all drivers, has
a variety of easily measurable quantities that may impact the utility of a
driver, such as travel time, distance, average speed, number of intersections
and so on, but how these quantities map to utility depends on the internal
preferences of the drivers.

\newcommand{\strategy}{\sigma^{\Gamesym}}
\newcommand{\simplex}[1]{\Delta_{#1}}
We model the players using a \defword{joint strategy},
$\strategy\in\simplex{\outcomes}$, which is a distribution over the game's
outcomes.  Coordination between players can exist, thus, this distribution need
not factor into independent strategies for each player.  Conceptually, a
trusted signaling mechanism, such as a traffic light, can be thought to sample
an outcome from $\strategy$ and communicate to each player
$\outcome_{\playeri}$, its portion of the joint-action.  Even in situations
where players are incapable of communication prior to play, correlated play is
attainable through repetition.  In particular, there are simple learning
dynamics that, when employed by each player independently, converge to a
correlated solution~\citep{foster96,hart2000}.

\newcommand{\pideviation}{f_{\playeri}}
\newcommand{\pideviationf}[1]{f_{\playeri}(#1)}
\newcommand{\devset}{\Phi}
\newcommand{\devsetf}{\Phi_{\deviation}}
\newcommand{\devswap}{\Phi^{\mbox{swap}}} 
If a player can benefit through a unilateral deviation from the proposed joint
strategy, the strategy is unstable.  As we are considering coordinated
strategies, a player may condition its deviations on the recommended action.
That is, a \defword{deviation for player $i$} is a function
$\func{\pideviation}{\actionsi}{\actionsi}$~\citep{agt4}.  To ease the notation,
we overload $\func{\pideviation}{\outcomes}{\outcomes}$ to be the function that
modifies only player $\playeri$'s action according to $\pideviation$.

\newcommand{\xsym}{x}
\newcommand{\ysym}{y}
\newcommand{\switchixy}{\operatorname{switch}^{\xsym\rightarrow\ysym}_{\playeri}}
\newcommand{\actioni}{\outcome_{\playeri}}
\newcommand{\switchixyf}[1]{\switchixy(#1)}
\newcommand{\devinternal}{\Phi^{\mbox{int}}}
\newcommand{\fixediy}{\operatorname{fixed}^{\rightarrow\ysym}_{\playeri}}
\newcommand{\fixediyaf}[1]{\fixediy(#1)}
\newcommand{\devexternal}{\Phi^{\mbox{ext}}} 
Two well-studied classes of deviations are the switch deviation, 
\begin{align} 
\switchixyf{\actioni} &= \left\{\begin{array}{cl} 
y & \mbox{if $\actioni = \xsym$}\\ 
\actioni & \mbox{otherwise},
\end{array}\right.
\end{align}
which substitutes one action for another, and the fixed deviation,
\begin{align}
\fixediyaf{\actioni} &= y,
\end{align}
which does not condition its change on the prescribed action.  A
\defword{deviation set}, denoted $\devset$, is a set of deviation functions.
We call the set of all switch deviations the internal deviation set,
$\devinternal$, and the set of all fixed deviations the external deviation set,
$\devexternal$.  The set $\devswap$ is the set of all deterministic deviations.
Given that the other players indeed play their recommended actions, there is no
strategic advantage to considering randomized deviations.

\newcommand{\regretf}[3]{r^{#1}_{#2}(#3)}
\newcommand{\regretw}[4]{r^{#1}_{#2}(#3|#4)}
The benefit of applying deviation $\pideviation$ when the players jointly play
$\outcome$ is known as \defword{instantaneous regret}.  We write the
instantaneous regret features as
\begin{align}
\regretf{\Gamesym}{\pideviation}{\outcome} &= \utilityif{\Gamesym}{\pideviationf{\outcome}} - \utilityif{\Gamesym}{\outcome},
\end{align}
and the instantaneous regret under utility function $\wsym$ as
\begin{align}
\regretw{\Gamesym}{\pideviation}{\outcome}{\wsym} &= \utilityiw{\Gamesym}{\pideviationf{\outcome}}{\wsym} - \utilityiw{\Gamesym}{\outcome}{\wsym} =
\inner{\regretf{\Gamesym}{\pideviation}{\outcome}}{\wsym}.
\end{align}

\newcommand{\deviation}{f} 
More generally, we can consider broader classes of deviations than the two we
have mentioned.  Conceptually, a deviation is a strategy modification and its
regret is its benefit to a particular player.  As we will ultimately only work
with the regret features, we can now suppress the implementation details while
bearing in mind that a deviation typically has these prescribed semantics.
That is, a \defword{deviation} $\deviation\in\devset$ has associated
instantaneous regret features, $\regretf{\Gamesym}{\deviation}{\cdot}$, and
instantaneous regret, $\regretw{\Gamesym}{\deviation}{\cdot}{\wsym}$.

\newcommand{\expectation}[2]{\mathbb{E}_{#1}\left[#2\right]}
\newcommand{\distributed}{\sim} 
As a player is only privileged to its own portion of the coordinated outcome,
it must reason about its \defword{expected regret}.  We write the expected
regret features as 
\begin{align}
\regretf{\Gamesym}{\deviation}{\strategy} &= \expectation{\outcome\distributed\strategy}{\regretf{\Gamesym}{\deviation}{\outcome}},
\end{align}
and the expected regret under utility function $\wsym$ as
\begin{align} 
\regretw{\Gamesym}{\deviation}{\strategy}{\wsym} &= \expectation{\outcome\distributed\strategy}{\regretw{\Gamesym}{\deviation}{\outcome}{\wsym}}
= \inner{\regretf{\Gamesym}{\deviation}{\strategy}}{\wsym}.
\end{align}

\newcommand{\Regret}[4]{\operatorname{Regret}^{#1}_{#2}(#3|#4)}
\newcommand{\eqmeps}{\varepsilon} 
A joint strategy is in \defword{equilibrium} or, in a sense, stable if no
player can benefit through a unilateral deviation.  We can quantify this
stability using expected regret with respect to the deviation set $\devset$,
\begin{align}
\Regret{\Gamesym}{\devset}{\strategy}{\wsym} = \max_{\deviation\in\devset}\regretw{\Gamesym}{\deviation}{\strategy}{\wsym},
\end{align}
and call a joint strategy $\strategy$ an \defword{$\eqmeps$-equilibrium} if 
\begin{align}
\Regret{\Gamesym}{\devset}{\strategy}{\wsym} \le \eqmeps.  
\end{align}

The most general deviation set, $\devswap$, corresponds with the
\defword{$\varepsilon$-correlated equilibrium} solution
concept~\citep{osborne1994course,agt4}.  Thus, regret can be thought of as the
natural substitute for utility when assessing the optimality of behavior in
multi-agent settings.

The set $\devswap$ is typically intractably large.  Fortunately, internal
regret closely approximates swap regret and is polynomially-sized in both the
number of actions and players.  
\begin{lemma}
If joint strategy $\strategy$ has $\eqmeps$ internal regret, then it is an
$\numactions\eqmeps$-correlated equilibrium under utility function $w$.  That
is, $\forall \wsym\in\Vsym$, 
\begin{align}
\Regret{\Gamesym}{\devinternal}{\strategy}{\wsym} & \le \Regret{\Gamesym}{\devswap}{\strategy}{\wsym} \le
\numactions\cdot\Regret{\Gamesym}{\devinternal}{\strategy}{\wsym}.
\end{align}
\label{lemma:internalapproxswap}
\end{lemma}
The proof is provided in the Appendix.

\section{Behavior Estimation in a Matrix Game}

We are now equipped with the tools necessary to introduce our approach for
imitation learning in multi-agent settings.  We start by assuming a notion of
rationality on the part of the game's players.  By leveraging this assumption,
we will then derive an estimation procedure with much better statistical
properties than methods that are unaware of the game's structure.

\subsection{Rationality and the ICE Polytope}

\newcommand{\obsidxsym}{t}
\newcommand{\outcomet}{\outcome^{\obsidxsym}}
\newcommand{\Tsym}{T}
\newcommand{\truestrategy}{\strategy}
\newcommand{\demonstrategy}{\tilde{\sigma}^{\Gamesym}}
Let $\{\outcomet\}_{\obsidxsym=1}^{\Tsym}$ be a sequence of $\Tsym$ independent
observations of behavior in game $\Gamesym$ distributed according to
$\truestrategy$, the players' \defword{true behavior}.  We call the empirical
distribution of the observations, $\demonstrategy$, the \defword{demonstrated
behavior}.

\newcommand{\predstrategy}{\hat{\sigma}^{\Gamesym}}
\newcommand{\predstrategya}[1]{\hat{\sigma}^{\Gamesym}(#1)}
We aim to learn a distribution $\predstrategy$, called the \defword{predicted
behavior}, an estimation of the true behavior from these demonstrations.
Moreover, we would like our learning procedure to extract the motives
for the behavior so that we may imitate the players in similarly
structured, but unobserved games.  Initially, let us consider just the
estimation problem.  While deriving our method, we will assume we have access
to the players' true behavior.  Afterwards, we will analyze the error
introduced by approximating from the demonstrations.

Imitation appears hard barring further assumptions.  In particular, if the
agents are unmotivated or their intentions are not coerced by the observed
game, there is little hope of recovering principled behavior in a new game.
Thus, we require a form of rationality.
\newcommand{\strategyprime}{\acute{\sigma}^{\Gamesym}}
\begin{proposition}
The players in a game are \defword{rational} with respect to deviation set
$\devset$ if they prefer joint-strategy $\strategy$ over joint strategies
$\strategyprime$ when
\begin{align}
\Regret{\Gamesym}{\devset}{\strategy}{\wstarsym} & < \Regret{\Gamesym}{\devset}{\strategyprime}{\wstarsym}.
\end{align}
\end{proposition}
Our rationality assumption states that the players are driven to minimize their
regret.  It is not necessarily the case that they indeed have low or no regret,
but simply that they can evaluate their preferences and that they prefer joint
strategies with low regret.  Through this assumption, we will be able to reason
about the players' behavior solely through the game's features; this is what
leads to the improved statistical properties of our approach.

As agents' true preferences $\wstarsym$ are unknown, we consider an
encompassing assumption that requires that estimated behavior satisfy this
property for all possible utility weights.  A prediction $\predstrategy$ is
\defword{strongly rational} with respect to deviation set $\devset$ if
\begin{align}
\forall \wsym\in\Vsym,\;\;\Regret{\Gamesym}{\devset}{\predstrategy}{\wsym}& \le \Regret{\Gamesym}{\devset}{\truestrategy}{\wsym}.
\end{align}

This assumption is similar in spirit to the utility matching assumption
employed by inverse optimal control techniques in single-agent settings.  As in
those settings, we have an {\em if and only if} guarantee relating rationality
and strong rationality~\citep{abbeel2004,ziebart2008}. 
\begin{theorem}
If a prediction $\predstrategy$ is strongly rational with respect to deviation
set $\devset$ and the players are rational with respect to $\devset$, then they
do not prefer $\truestrategy$ over $\predstrategy$.
\end{theorem}
This is immediate as $\wstarsym\in\Vsym$.

Phrased another way, a strongly rational prediction is no worse than the true behavior.
\begin{corollary}
If a prediction $\predstrategy$ is strongly rational with respect to deviation
set $\devset$ and the true behavior is an $\eqmeps$-equilibrium with respect to
$\devset$ under utility function $\wstarsym\in\Vsym$, then $\predstrategy$ is
also an $\eqmeps$-equilibrium.
\end{corollary}
Again, the proof is immediate as
$\Regret{\Gamesym}{\devset}{\predstrategy}{\wstarsym} \le
\Regret{\Gamesym}{\devset}{\truestrategy}{\wstarsym} \le \eqmeps$.

Conversely, if we are uncertain about the true utility function we {\bf must} assume
strong rationality or we risk predicting less desirable behavior.
\begin{theorem}
If a prediction $\predstrategy$ is not strongly rational with respect to
deviation set $\devset$ and the players are rational, then there exists a
$\wstarsym\in\Vsym$ such that $\truestrategy$ is preferred to $\predstrategy$.
\end{theorem}
The proof follows from the negation of the definition of strong rationality.

By restricting our attention to strongly rational behavior, at worst agents
acting according to their unknown true preferences will be indifferent between
our predictive distribution and their true behavior.  That is, strong
rationality is necessary and sufficient under the assumption players are
rational given no knowledge of their true utility function.

\newcommand{\devdist}{\eta}
\newcommand{\devdistf}{\eta_{\deviation}}
Unfortunately, a direct translation of the strong rationality requirement into
constraints on the distribution $\predstrategy$ leads to a non-convex
optimization problem as it involves products of varying utility vectors with
the behavior to be estimated.  Fortunately, we can provide an equivalent
concise convex description of the constraints on $\predstrategy$ that ensures
any feasible distribution satisfies strong rationality. We denote this set of
equivalent constraints as the {\em Inverse Correlated Equilibria} (ICE)
polytope.
\begin{definition}[Standard ICE Polytope]
\begin{align}
\regretf{\Gamesym}{\deviation}{\predstrategy} & = \expectation{g\distributed\devdistf}{\regretf{\Gamesym}{g}{\truestrategy}},&\forall \deviation\in\devset \\ 
\devdistf&\in\simplex{\devsetf}, &\forall \deviation\in\devset\\ 
\predstrategy&\in\simplex{\outcomes}.
\end{align}
\end{definition}
Here, we have introduced $\devsetf$, the set of deviations that $\deviation$
will be compared against.  Our rationality assumption corresponds to when
$\devsetf=\devset$, but there are different choices that have reasonable
interpretations as alternative rationality assumptions.  For example, if each
switch deviation is compared only against switches for the same player---a
more restrictive condition---then the quality of the equilibrium is
measured by the {\em sum} of all players' regrets, as opposed to only the
one with the most regret.

The following corollary equates strong rationality and the standard ICE
polytope.  
\begin{corollary}
A prediction $\predstrategy$ is strongly rational with respect to deviation set
$\devset$ if and only if for all $\deviation\in\devset$ there exists
$\devdistf\in\simplex{\devset}$ such that $\predstrategy$ and $\devdist$
satisfy the standard ICE polytope.
\label{c:standardice}
\end{corollary}

We now show a more general result that implies Corollary~\ref{c:standardice}.
We start by generalizing the notion of strong rationality by restricting
$\wstarsym$ to be in a known set $\Ksym\subseteq\Vsym$.  We say a prediction
$\predstrategy$ is $\Ksym$-strongly rational with respect to deviation set
$\devset$ if
\begin{align} \forall
\wsym\in\Ksym,\;\;\Regret{\Gamesym}{\devset}{\predstrategy}{\wsym}& \le
\Regret{\Gamesym}{\devset}{\truestrategy}{\wsym}.
\end{align}
If $\Ksym$ is convex with non-empty relative interior and $0\in\Ksym$, we
derive the $\Ksym$-ICE polytope.
\begin{definition}[$\Ksym$-ICE Polytope]
\begin{align}
\regretf{\Gamesym}{\deviation}{\predstrategy} - \expectation{g\distributed\devdistf}{\regretf{\Gamesym}{g}{\truestrategy}}& \in -\Kdualsym, &\forall \deviation\in\devset \\
\devdistf&\in\simplex{\devsetf}, &\forall \deviation\in\devset\\
\predstrategy&\in\simplex{\outcomes}. 
\end{align}
\end{definition}
Note that the above constraints are linear in $\predstrategy$ and
$\devdistf$, and $\Kdualsym$, the dual cone, is convex.  The following theorem
shows the equivalence of the $\Ksym$-ICE polytope and $\Ksym$-strong
rationality.
\begin{theorem} 
A prediction $\predstrategy$ is $\Ksym$-strongly rational with respect to
deviation set $\devset$ if and only if for all $\deviation\in\devset$ there
exists $\devdistf\in\simplex{\devset}$ such that $\predstrategy$ and $\devdist$
satisfy the $\Ksym$-ICE polytope.
\label{thm:kice}
\end{theorem}
The proof is provided in the Appendix.

\newcommand{\RKplussym}{\Rsym^{K}_{+}}
By choosing $\Ksym=\Vsym$, then $\Kdualsym=\{0\}$ and the polytope reduces to
the standard ICE polytope.  Thus, Corollary~\ref{c:standardice} follows
directly from Theorem~\ref{thm:kice}.  By choosing $\Ksym$ to be the positive
orthant, $\Ksym=\Kdualsym=\RKplussym$, the polytope reduces to the following
inequalities.  Here, we explicitly assume the utility to be a positive linear function
of the features. 
\begin{definition}[Positive ICE	Polytope]
\begin{align}
\regretf{\Gamesym}{\deviation}{\predstrategy} & \le \expectation{g\distributed\devdistf}{\regretf{\Gamesym}{g}{\truestrategy}}, &\forall \deviation\in\devset \\
\devdistf&\in\simplex{\devsetf}, &\forall \deviation\in\devset\\
\predstrategy&\in\simplex{\outcomes}.
\end{align}
\end{definition}

Predictive behavior within the ICE polytope will retain the quality of the
demonstrations provided.  The following corollaries formalize this guarantee.
\begin{corollary}
If the true behavior is an $\eqmeps$-correlated equilibrium under $\wstarsym$
in game $\Gamesym$, then a prediction $\predstrategy$ that satisfies the
standard ICE polytope where $\devset = \devswap$ and $\forall
\deviation\in\devset, \devsetf=\devset$ is also an $\eqmeps$-correlated
equilibrium.
\end{corollary} 
This follows immediately from the definition of an approximate correlated
equilibrium.

\begin{corollary} 
If the true behavior is an $\eqmeps$-correlated equilibrium under $\wstarsym$
in game $\Gamesym$, then a prediction $\predstrategy$ that satisfies the
standard ICE polytope where $\devset = \devinternal$ and $\forall
\deviation\in\devset, \devsetf=\devset$ is also an
$\numactions\eqmeps$-correlated equilibrium.
\end{corollary}
This follows immediately from Lemma~\ref{lemma:internalapproxswap}.

\newcommand{\strategyi}{\sigma^{\Gamesym}_{\playeri}} In two-player
constant-sum games, we can make stronger statements about our predictive
behavior.  In particular, when these requirements are satisfied we may reason
about games without coordination.  That is, each player chooses their action
independently using their \defword{strategy}, $\strategyi$ a distribution over
$\actionsi$.  A \defword{strategy profile} $\strategy$ consists of a strategy
for each player.  It defines a joint-strategy with no coordination
between the players.

A game is constant-sum if there is a fixed amount of utility divided among the
players.  That is, if there is a constant $C$ such that
$\forall\outcome\in\outcomes$,
\begin{align}
\sum_{\playeri\in\players}\utilityiw{\Gamesym}{\outcome}{\wstarsym} & = C.
\end{align}

In settings where the players act independently, we use external regret to
measure a profile's stability, which corresponds with the famous \defword{Nash
equilibrium} solution concept~\citep{osborne1994course}.  By using the ICE
polytope with external regret, we can recover a Nash equilibrium if one is
demonstrated in a constant-sum game.  
\begin{theorem}
If the true behavior is an $\eqmeps$-Nash equilibrium in a two-player
constant-sum game $\Gamesym$, then the marginal strategies formed from a
prediction $\predstrategy$ that satisfies the standard ICE polytope where
$\devset = \devexternal$ and $\forall \deviation\in\devset, \devsetf=\devset$
is a $2\eqmeps$-Nash equilibrium. 
\label{thm:nash} 
\end{theorem}
The proof is provided in the Appendix.

In general, there can be infinitely many correlated equilibrium with vastly
different properties.  One such property that has received much attention is
the social welfare of a joint strategy, which refers to the total utility over
all players.  Our strong rationality assumption states that the players have no
preference on which correlated equilibrium is selected, and thus without
modification cannot capture such a concept should it be demonstrated.  We can
easily maintain the social welfare of the demonstrations by additionally
preserving the players' utilities along side the constraints prescribed by the
ICE polytope. A joint strategy is utility-preserving under all utility
functions if
\begin{align}
\forall \wsym\in\Vsym,\playeri\in\players,\;\;\utilityiw{\Gamesym}{\predstrategy}{\wsym} & = \utilityiw{\Gamesym}{\truestrategy}{\wsym}.
\end{align}
As with the
correspondence between strong rationality and the ICE polytope, utility
preservation can be represented as a set of linear equality constraints.  These
utility feature matching constraints are exactly the basis of many methods of
inverse optimal control~\citep{abbeel2004,ziebart2008}.  
\begin{theorem}
A joint strategy is utility-preserving under all utility functions if and only
if
\begin{align}
\playeri\in\players,\;\;\utilityif{\Gamesym}{\predstrategy} & = \utilityif{\Gamesym}{\truestrategy}.
\end{align}
\end{theorem}
The proof is due to~\citet{abbeel2004}.

A notable choice for $\devsetf$ is we compare each deviation only to itself.
As a consequence this enforces a stronger constraint that the regret under each
deviation, and in turn the overall regret, is the same under our prediction and
the demonstrations.  That is, $\predstrategy$ is \defword{regret-matching} as
for all $\wsym\in\Vsym$,
\begin{align}
\Regret{\Gamesym}{\devset}{\predstrategy}{\wsym} & = \Regret{\Gamesym}{\devset}{\truestrategy}{\wsym}.
\end{align}
Thus, regret-matching preserves the equilibrium qualities of the
demonstrations.

Unlike the correspondence between the ICE polytope and strong rationality,
matching the regret features for each deviation is \textbf{not} required for a
strategy to match the regrets of the demonstrations.  That is, the converse
does not hold.\footnote{We may sketch a simple counterexample.  Consider a game
with one player and three actions, $x$, $y$ and $y'$, where the utility for
playing $x$ is zero, and the utility for playing either $y$ or $y'$ is one.  If
the true behavior always plays $y$, then matching the regret features will
force the prediction to also play $y$.  Predicting $y'$ also matches the
regret, though.}
\begin{theorem}
A prediction $\predstrategy$ matches the regret of $\truestrategy$ for all
$\wsym\in\Vsym$ does not necessarily match the regret features of
$\truestrategy$.  
\end{theorem}

We use both utility and regret matching in our final set of experiments.  The
former for predictive reasons, the latter to allow for the use of smooth
minimization techniques.

\subsection{The Principle of Maximum Entropy}

As we are interested in the problem of statistical prediction of strategic
behavior, we must find a mechanism to resolve the ambiguity remaining after
accounting for the rationality constraints. The \defword{principle of maximum
entropy}, due to~\citet{jaynes1957}, provides a well-justified method for
choosing such a distribution.  This choice leads to not only statistical
guarantees on the resulting predictions, but to efficient optimization.

\newcommand{\entropy}[2]{H^{#1}(#2)}
\newcommand{\strategya}{\strategy(\outcome)}
\newcommand{\truestrategya}{\strategy(\outcome)}
The \defword{Shannon entropy} of a joint-strategy $\strategy$ is
\begin{align}
\entropy{\Gamesym}{\strategy} & = \expectation{\outcome\distributed\strategy}{-\log \strategya},
\end{align}
and the principle of maximum entropy advocates choosing the distribution with
maximum entropy subject to known constraints~\citep{jaynes1957}.  That is, 
\begin{align}
\sigma_{\text{MaxEnt}} & = \argmax_{\strategy\in\simplex{\outcomes}} \entropy{\Gamesym}{\strategy}, \quad\mbox{subject to:} \\ 
		       & \; g(\strategy) = 0 \text{ and } h(\strategy) \leq 0.  
\end{align} 
The constraint functions, $g$ and $h$, are typically chosen to capture the
important or most salient characteristics of the distribution.  When those
functions are affine and convex respectively, finding this distribution is a
convex optimization problem.  The resulting log-linear family of distributions
({\em e.g.}, logistic regression, Markov random fields, conditional random
fields) are widely used within statistical machine learning.

In the context of multi-agent behavior, the principle of maximum entropy has
been employed to obtain correlated equilibria with predictive guarantees in
normal-form games when the utilities are known {\em a priori}~\citep{ortiz2007}.
We will now leverage its power with our rationality assumption to select
predictive distributions in games where the utilities are unknown, but the
important features that define them are available.

\newcommand{\maximize}{\operatornamewithlimits{maximize}}
For our problem, the constraints are precisely that the distribution is in the
ICE polytope, ensuring that whatever we predict has no more regret than the
demonstrated behavior.  
\begin{definition} 
The primal maximum entropy ICE optimization problem is 
\begin{align} 
\maximize_{\predstrategy,\devdist} &\;\; \entropy{\Gamesym}{\predstrategy} \quad\mbox{subject to:}\\
\regretf{\Gamesym}{\deviation}{\predstrategy} - \expectation{g\distributed\devdistf}{\regretf{\Gamesym}{g}{\truestrategy}} & \in -\Kdualsym, &\forall \deviation\in\devset \\
\devdistf&\in\simplex{\devsetf}, &\forall \deviation\in\devset\\
\predstrategy&\in\simplex{\outcomes}.  
\end{align} 
\end{definition} 
This program is convex, feasible, and bounded.  That is, it has a solution and
is efficiently solvable using simple techniques in this form.

Importantly, the maximum entropy prediction enjoys the following guarantee:
\begin{lemma} 
\label{lem:maxent} 
The maximum entropy ICE distribution minimizes over all strongly rational
distributions the worst-case log-loss,
$\expectation{\outcome\distributed\strategy}{-\log_2
\predstrategya{\outcome}}$, when $\truestrategy$ is chosen adversarially but
subject to strong rationality.  
\end{lemma} 
The proof of Lemma~\ref{lem:maxent} follows immediately from the result of~\citet{grunwald2003}.

\subsection{Dual Optimization}

In this section, we will derive and describe a procedure for optimizing the
dual program for solving the MaxEnt ICE optimization problem.  We will see that
the dual multipliers can be interpreted as utility vectors and that
optimization in the dual has computational advantages.  We begin by presenting
the dual program.

\newcommand{\minimize}{\operatornamewithlimits{minimize}}
\newcommand{\dualweights}{\theta} \newcommand{\optdualweights}{\theta^*}
\newcommand{\dualweightf}{\theta_{\deviation}}
\newcommand{\optdualweightf}{\theta^*_{\deviation}}
\newcommand{\Zfunc}[1]{Z^{#1}(\dualweights)}
\newcommand{\logZfunc}[1]{\operatorname{logZ}^{#1}(\dualweights)}
\newcommand{\Kddualsym}{\Ksym^{**}}
\begin{theorem} 
The dual maximum entropy ICE optimization problem is the following non-smooth,
but convex program: 
\begin{align} 
\minimize_{\dualweightf\in\Kddualsym} & \;\; \sum_{\deviation\in\devset}\Regret{\Gamesym}{\devsetf}{\truestrategy}{\dualweightf} + \log\Zfunc{\Gamesym}, \mbox{~where} \\ 
\Zfunc{\Gamesym} & = \sum_{\outcome\in\outcomes}\exp\left(-\sum_{\deviation\in\devset}\regretw{\Gamesym}{\deviation}{\outcome}{\dualweightf} \right).  
\end{align} 
\label{thm:dual} 
\end{theorem} 
We derive the dual in the Appendix.

As the dual's feasible set has non-empty relative interior, strong duality
holds by Slater's condition---there is no duality gap.  We can also use a
dual solution to recover $\predstrategy$.
\begin{lemma}
Strong duality holds for the maximum entropy ICE optimization problem and given
optimal dual weights $\optdualweights$, the maximum entropy ICE joint-strategy
$\predstrategy$ is 
\begin{align} 
\predstrategya{\outcome} & \propto \exp\left(-\sum_{\deviation\in\devset}\regretw{\Gamesym}{\deviation}{\outcome}{\optdualweightf} \right).
\label{eqn:dualtoprimal} 
\end{align} 
\end{lemma}

\newcommand{\gsym}{g}
\newcommand{\gdeviation}{g^f}
\newcommand{\deviationstar}{f^*}
\newcommand{\deviationprime}{f'}
\begin{algorithm}[tb]
\caption{Dual MaxEnt ICE Gradient}
\label{alg:dualgradient} 
\begin{algorithmic} 
\STATE {\bfseries Input:} Let $\predstrategy$ be the prediction given the
current dual weights, $\dualweights$, as from
Equation~\eqref{eqn:dualtoprimal}.  
\FOR{$\deviation \in \devset$} 
\STATE $\deviationstar \leftarrow \argmax_{\deviationprime\in\devsetf}\regretw{\Gamesym}{\deviationprime}{\truestrategy}{\dualweightf}$
\STATE $\gdeviation \leftarrow \regretf{\Gamesym}{\deviationstar}{\truestrategy} - \regretf{\Gamesym}{\deviation}{\predstrategy}$ 
\ENDFOR 
\STATE {\bfseries return} $\gsym$ 
\end{algorithmic} 
\label{alg:dual}
\end{algorithm}

The dual formulation of our program has important inherent computational
advantages.  First, so long as $\Ksym$ is simple, the optimization is
particularly well-suited for gradient-based optimization, a trait not shared by
the primal program.  Second, the number of dual variables,
$\abs{\devset}\dim{\Vsym}$, is typically much fewer than the number of primal
variables, $\abs{\outcomes}+\abs{\devset}^2$.  Though the work per iteration is
still a function of $\abs{\outcomes}$ (to compute the partition function),
these two advantages together let us scale to larger problems than if we
consider optimizing the primal objective. Computing the expectations necessary
to descend the dual gradient can leverage recent advances in the structured,
compact game representations: in particular, any graphical game with
low-treewidth or finite horizon Markov game~\citep{kakade2003correlated} enables
these computations to be performed in time that scales only polynomially in the
number of decision makers.

Algorithm~\ref{alg:dual} describes the dual gradient computation.  This can be
incorporated with any non-smooth gradient method, such as the projected
subgradient method~\citep{shor1985}, to approach the optimum dual weights.

\section{Behavior Estimation in Parameterized Matrix Games}

\newcommand{\Gamessym}{\mathcal G}
To account for stochastic, or varying environments, we now consider {\em
distributions over} games.  For example, rain may affect travel time along some
routes and make certain modes of transportation less desirable, or even
unavailable.  Operationally, nature samples a game prior to play from a
distribution known to the players.  The players then as a group determine a
joint strategy conditioned on the particular game and an outcome is drawn by a
coordination device.  We let $\Gamessym$ denote our class of games.

\newcommand{\Gamet}{\Gamesym^{\obsidxsym}}
\newcommand{\truegamedist}{\xi}
\newcommand{\truestrategies}{\sigma}
\newcommand{\demongamedist}{\tilde{\xi}}
\newcommand{\demonstrategies}{\tilde{\sigma}}
As before, we observe a sequence of $\Tsym$ independent observations of play,
but now in addition to an outcome we also observe nature's choice at each time
$\obsidxsym$.  Let $\{(\Gamet,\outcomet)\}_{\obsidxsym=1}^\Tsym$ be the
aforementioned sequence of observations drawn from $\truegamedist$ and
$\truestrategies$, the \defword{true behavior}.   The empirical distribution of
the observations, $\demongamedist$ and $\demonstrategies$, together are the
\defword{demonstrated behavior}.

\newcommand{\predstrategies}{\hat{\sigma}} 
Now we aim to learn a \defword{predictive behavior} distribution,
$\predstrategy$, for {\em any} $\Gamesym\in\Gamessym$, even ones we have not
yet observed.  Clearly, we must leverage the observations across the entire
family to achieve good predictive accuracy.  We continue to assume that the
players' utility is an unknown linear function, $\wstarsym$, of the games'
features and that this function is fixed across $\Gamessym$.  Next, we amend
our notion of regret and our rationality assumption.

\subsection{Behavior Estimation through Conditional ICE}

Ultimately, we wish to simply employ an additional expectation over the game
distribution when reasoning about the regret and regret features.  To do this,
our notion of a deviation needs to account for the fact that it may be executed
in games with different structures.  Operationally, one way to achieve this
is by having a deviation not act when it is applied to such a game, which
increases the size of $\devset$ by a factor of $\abs{\Gamessym}$.  If the
actions, and in turn the deviations, have similar semantic meanings across our
entire family of games, one can simply share the deviations across all games.
This allows for one to achieve transfer over an infinitely large class.  Given
such a decision, we write the expected regret features under deviation
$\deviation$ as 
\begin{align}
\regretf{\truegamedist}{\deviation}{\truestrategies} &= \expectation{\Gamesym\distributed\truegamedist}{\regretf{\Gamesym}{\deviation}{\truestrategy}},
\end{align} 
and the expected regret under utility function $\wsym$ as
\begin{align} 
\regretw{\truegamedist}{\deviation}{\truestrategies}{\wsym} &= \expectation{\Gamesym\distributed\truegamedist}{\regretw{\Gamesym}{\deviation}{\truestrategy}{\wsym}}.
\end{align} 
Again, we quantify the stability of a set of joint strategies using
this new notion of expected regret with respect to the deviation set $\devset$,
\begin{align} 
\Regret{\truegamedist}{\devset}{\truestrategies}{\wsym} &= \max_{\deviation\in\devset}\regretw{\truegamedist}{\deviation}{\truestrategies}{\wsym},
\end{align} 
which, in turn, entails a notion of an $\eqmeps$-equilibrium for a
set of joint strategies, a modified rationality assumption, and a slight
modification to the $\Ksym$-ICE polytope, 
\begin{definition}[Conditional $\Ksym$-ICE Polytope] 
\begin{align}
\regretf{\truegamedist}{\deviation}{\predstrategies} - \expectation{g\distributed\devdistf}{\regretf{\truegamedist}{g}{\truestrategies}} & \in -\Kdualsym, &\forall \deviation\in\devset \\
\devdistf&\in\simplex{\devsetf}, &\forall \deviation\in\devset\\
\predstrategy&\in\simplex{\outcomes}.&\forall \Gamesym\in\Gamessym
\end{align} 
\end{definition}

All that remains is to adjust our notion of entropy to take into account a
distribution over games.  In particular, we choose to maximize the expected
entropy of our prediction, which is conditioned on the game sampled by chance.
\begin{definition} 
The conditional Shannon entropy of a set of strategies $\truestrategies$ when
games are distributed according to $\truegamedist$ is 
\begin{align}
\entropy{\truegamedist}{\truestrategies} & = \expectation{\Gamesym\distributed\truegamedist}{\entropy{\Gamesym}{\truestrategy}}.
\end{align} 
\end{definition}

The modified dual optimization problem has a familiar form.  We now use the new
notion of regret and take the expected value of the log partition function.
\begin{theorem} 
The dual conditional maximum entropy ICE optimization problem
is 
\begin{align} 
\minimize_{\dualweightf\in\Kddualsym} & \;\; \sum_{\deviation\in\devset}\Regret{\truegamedist}{\devsetf}{\truestrategies}{\dualweightf} + \expectation{\Gamesym\distributed\truegamedist}{\log\Zfunc{\Gamesym}}.
\end{align} 
\end{theorem} 
To recover the predicted behavior for a particular game, we use the same
exponential family form as before.  

As with any machine learning technique, it is advisable to employ some form of
complexity control on the resulting predictor to prevent over-fitting.  As we
now wish to generalize to unobserved games, we too should take the appropriate
precautions.  In our experiments, we employ $L1$ and $L2$ regularization terms
to the dual objective for this purpose.  Regularization of the dual weights
effectively alters the primal constraints by allowing them to hold
approximately, leading to higher entropy solutions~\citep{dudik07}.

\subsection{Behavior Transfer without common deviations}

A principal justification of inverse optimal control techniques that attempt to
identify behavior in terms of utility functions is the ability to consider what
behavior might result if the underlying decision problem were changed while the
interpretation of features into utilities remain the
same~\citep{ng2000algorithms,ratliff2006}.  This enables prediction of agent
behavior in a no-regret or agnostic sense in problems such as a robot
encountering novel terrain~\citep{Silver_2010_6638} as well as route
recommendation for drivers traveling to unseen
destinations~\citep{Ziebart2008b}.  

Econometricians are interested in similar situations, but for much different
reasons.  Typically, they aim to validate a model of market behavior from
observations of product sales.  In these models, the firms assume a fixed
pricing policy given known demand.  The econometrician uses this fixed policy
along with product features and sales data to estimate or bound both the
consumers' utility functions as well as unknown production parameters, like
markup and production cost~\citep{berry95,nevo01}.  In this line of work, the
observed behavior is considered accurate to start with; it is unclear how
suitable these methods are for settings with limited or noisy observations.

\newcommand{\Gamepsym}{\acute{\Gamma}} 
In our prior work, we introduced an
approach to behavior transfer applicable between games with different action
sets~\citep{waugh11}.  It is based off the assumption of \defword{transfer
rationality}, or for two games $\Gamesym$ and $\Gamepsym$ and some constant
$\kappa > 0$, 
\begin{align} \forall \wsym\in\Vsym,\: \Regret{\Gamepsym}{\devset}{\predstrategy}{\wsym} \le \kappa\Regret{\Gamesym}{\devset}{\truestrategy}{\wsym}.  
\end{align} 
Roughly, we assume that under preferences with low regret in the original game,
the behavior in the unobserved game should also have low regret.  By enforcing
this property, if the agents are performing well with respect to their true
preferences, then the transferred behavior will also be of high quality.

Assuming transfer rationality is equivalent to using the conditional ICE
estimation program with differing game distributions for the predicted and
demonstrated regret features.  In such a case, the program is not necessarily
feasible and the constraints must be relaxed.  For example, a slack variable
may be added to the primal, or through regularization in the dual.  We note
that this requires the estimation program to be run at test time.

\section{Sample Complexity}

In practice, we do not have full access to the agents' true behavior---if we
did, prediction would be straightforward and we would not require our
estimation technique.  Instead, we may only approximate the desired
expectations by averaging over a finite number of observations,  
\begin{align}
\regretw{\demongamedist}{\deviation}{\demonstrategies}{\wsym} & \approx \frac{1}{T}\sum_{t=1}^T \regretw{\Gamet}{\deviation}{\outcomet}{\wsym}.
\end{align}
In real applications there are costs associated with gathering these
observations and, thus, there are inherent limitations on the quality of this
approximation.  Next, we will analyze the sensitivity of our approach to these
types of errors.

\newcommand{\maxregret}{\Delta} 
\newcommand{\maxregretf}[1]{\Delta(#1)} 
First, although $\abs{\outcomes}$ is exponential in the number of players, our
technique only accesses $\demonstrategies$ through expected regret features of
the form $\regretf{\demongamedist}{\deviation}{\demonstrategies}$.  That is, we
need only approximate these features accurately, not the distribution
$\truestrategies$.  For finite-dimensional vector spaces, we can bound how
well the regrets match in terms of $\abs{\devset}$ and the dimension of the
space.  
\begin{theorem} 
With probability at least $1-\delta$, for any $\wsym$, by observing $\Tsym \ge
\frac{1}{2\epsilon^2}\log\frac{2\abs{\devset}\dim{\Vsym}}{\delta}$ outcomes we
have for all deviations
$\regretw{\demongamedist}{\deviation}{\demonstrategies}{\wsym} \le
\regretw{\truegamedist}{\deviation}{\truestrategies}{\wsym} +
\epsilon\maxregret\norm{\wsym}_1$.  
\label{thm:samplefinite} 
\end{theorem}
where $\maxregret$ is the maximum possible regret over all basis
directions. The proof is an application of the union bound and Hoeffding's
inequality and is provided in the Appendix.

Alternatively, we can bound how well the regrets match independently of the
space's dimension by considering each utility function separately.  
\begin{theorem} 
With probability at least $1-\delta$, for any $\wsym$, by observing $\Tsym \ge
\frac{1}{2\epsilon^2}\log\frac{\abs{\devset}}{\delta}$ outcomes we have for all
deviations $\regretw{\demongamedist}{\deviation}{\demonstrategies}{\wsym} \le
\regretw{\truegamedist}{\deviation}{\truestrategies}{\wsym} +
\epsilon\maxregretf{\wsym}$.  
\label{thm:sampleinf} 
\end{theorem} 
where $\maxregretf{\wsym}$ is the maximum possible regret under $\wsym$.
Again, the proof is in the Appendix.

Both of the above bounds imply that, so long as the true utility function is
not too complex, with high probability we need only logarithmic many samples in
terms of $\abs{\devset}$ and $\dim{\Vsym}$ to closely approximate
$\regretf{\truegamedist}{\deviation}{\truestrategies}$ and avoid a large
violation of our rationality condition.  
\begin{theorem} 
If for all $\deviation$,
$\regretw{\demongamedist}{\deviation}{\demonstrategies}{\wsym} \le
\regretw{\truegamedist}{\deviation}{\truestrategies}{\wsym} + \gamma$, then
$\Regret{\demongamedist}{\devset}{\demonstrategies}{\wsym} \le
\Regret{\truegamedist}{\devset}{\truestrategies}{\wsym} + \gamma$.
\end{theorem} 
\begin{proof} 
For all deviations, $\deviation\in\devset$,
$\regretw{\demongamedist}{\deviation}{\demonstrategies}{\wsym} \le
\regretw{\truegamedist}{\deviation}{\truestrategies}{\wsym} + \gamma \le
\Regret{\truegamedist}{\devset}{\truestrategies}{\wsym} + \gamma$.  In
particular, this holds for the deviation that maximizes the demonstrated
regret.  
\end{proof}

\section{Experimental Results}

\subsection{Synthetic Routing Game}

To evaluate our approach experimentally, we first consider a simple synthetic
routing game.  Seven drivers in this game choose how to travel home during rush
hour after a long day at the office.  The different road segments have varying
capacities that make some of them more or less susceptible to congestion.  Upon
arrival home, the drivers record the total time and distance they traveled, the
fuel that they used, and the amount of time they spent stopped at intersections
or in congestion---their utility features.

In this game, each of the drivers chooses from four possible routes, yielding
over $16,000$ possible outcomes.  We obtained an $\eqmeps$-social welfare
maximizing correlated equilibrium for those drivers using a subgradient method
where the drivers preferred mainly to minimize their travel time, but were also
slightly concerned with fuel cost.  The demonstrated behavior $\demonstrategy$
was sampled from this true behavior distribution $\strategy$.

In Figure~\ref{fig:logloss} we compare the prediction accuracy of MaxEnt ICE,
measured using log loss, $\expectation{\outcome\distributed\strategy}{-\log_2
\predstrategya{\outcome}}$, against a number of baselines by varying the number
of observations sampled from the $\eqmeps$-equilibrium.  The baseline
algorithms are: a smoothed multinomial distribution over the joint-actions, a
logistic regression classifier parameterized with the outcome utilities, and a
maximum entropy inverse optimal control approach~\citep{ziebart2008} trained
individually for each player.

\begin{figure}[h] 
\begin{center}
\centerline{\includegraphics[width=.68\textwidth]{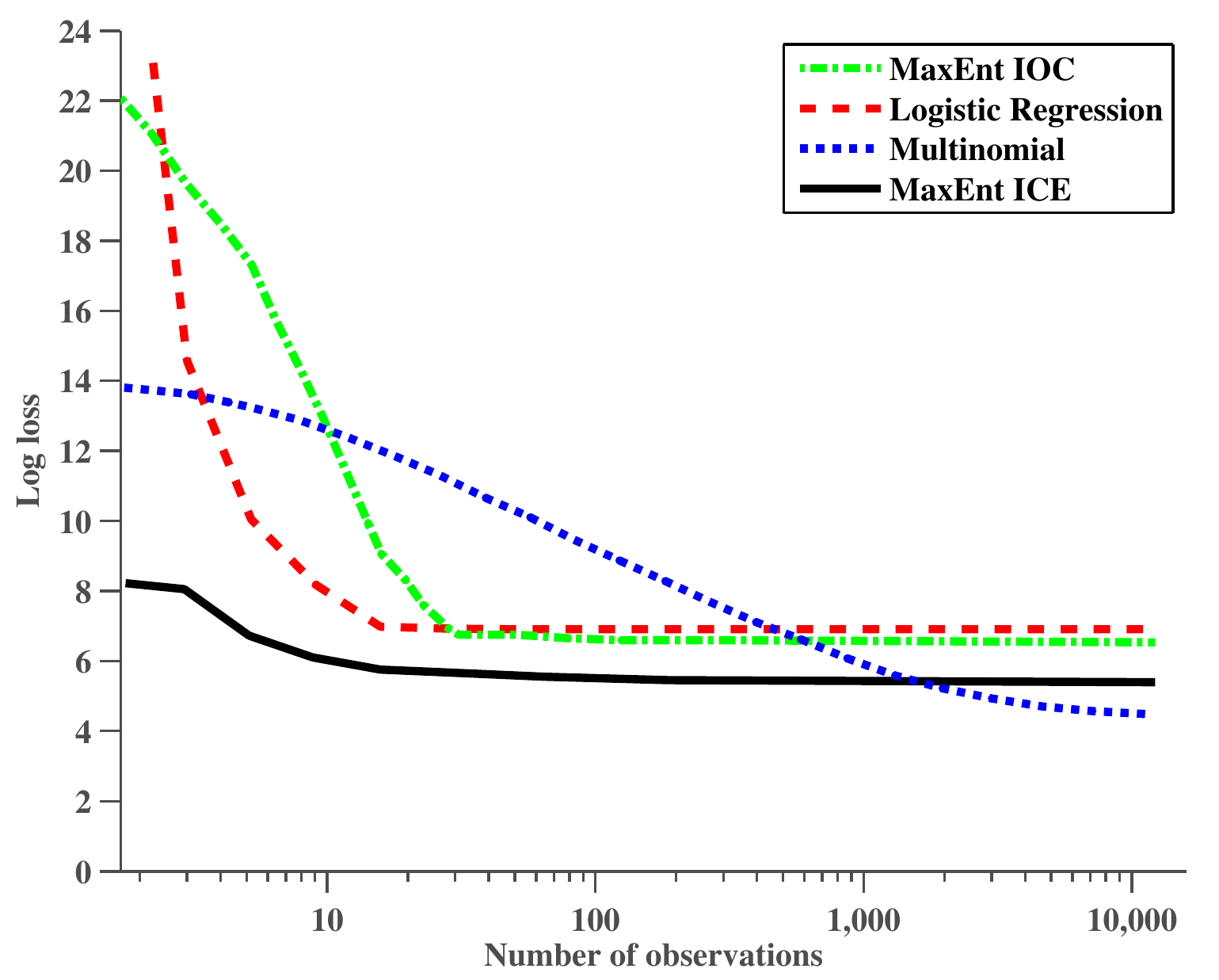}} 
\caption{Prediction error (log loss) as a function of number of observations in
the synthetic routing game.} 
\label{fig:logloss} 
\end{center} 
\vskip -0.2in
\end{figure}

In Figure~\ref{fig:logloss}, we see that MaxEnt ICE predicts behavior with
higher accuracy than all other algorithms when the number of observations is
limited.  In particular, it achieves close to its best performance with only
$16$ observations.  The maximum likelihood estimator eventually overtakes
it, as expected since it will ultimately converge to $\strategy$, but only
after $10,000$ observations, or close to as many observations as there are
outcomes in the game.  MaxEnt ICE cannot learn the true behavior exactly in
this case without additional constraints due to the social welfare criteria the
true behavior optimizes.  That is, our rationality assumption does not hold in
this case.  We note that the logistic regression classifier and the inverse
optimal control techniques perform better than the multinomial under low sample
sizes, but they fail to outperform MaxEnt ICE due to their inability to
appreciate the strategic nature of the game.

Next, we evaluate behavior transfer from this routing game to four similar
games, the results of which are displayed in Table~\ref{tbl:transfer}.  The
first game, {\em Add Highway}, adds a new route to the game.  That is, we
simulate the city building a new highway.  The second game, {\em Add Driver},
adds another driver to the game.  The third game, {\em Gas Shortage}, keeps the
structure of the game the same, but changes the reward function to make gas
mileage more important to the drivers.  The final game, {\em Congestion},
simulates adding construction to the major roadway, delaying the drivers.
Here, we do not share deviations across the training and test game and we add a
slack variable in the primal to ensure feasibility.

\begin{table}[t] 
\begin{center} 
\begin{small} 
\begin{sc} 
\begin{tabular}{lcc}
\hline 
Problem & Logistic Regression & MaxEnt Ice \\ \hline 
Add Highway & 4.177  & 3.093 \\ 
Add Driver & 4.060 & 3.477 \\ 
Gas Shortage & 3.498 & 3.137 \\ 
Congestion & 3.345 & 2.965 \\ \hline 
\end{tabular}
\end{sc} 
\end{small} 
\end{center} 
\caption{Transfer error (log loss) on unobserved games.} 
\label{tbl:transfer} 
\end{table}

These transfer experiments even more directly demonstrate the benefits of
learning utility weights rather than directly learning the joint-action
distribution; direct strategy-learning approaches are incapable of being
applied to general transfer setting.  Thus, we can only compare against the
Logistic Regression.  We see from Table~\ref{tbl:transfer} that MaxEnt ICE
outperforms the Logistic Regression in all of our tests.  For reference, in these
new games, the uniform strategy has a loss of approximately $6.8$ in all games,
and the true behavior has a loss of approximately $2.7$.

These experiment demonstrates that learning underlying utility functions to
estimate observed behavior can be much more data-efficient for small sample
sizes.  Additionally, it shows that the regret-based assumptions of MaxEnt ICE
are beneficial in strategic settings, even though our rationality assumption
does not hold in this case.

\subsection{Market Entry Game}

We next evaluate our approach against a number of baselines on data gathered
for the Market Entry Prediction Competition~\citep{erev2010}.  The game has four
players and is repeated for fifty trials and is meant to simulate a firm's
decision to enter into a market.  On each round, all four players simultaneous
decide whether or not to open a business.  All players who enter the market
receive a stochastic payoff centered at $10 - kE$, where $k$ is a fixed
parameter unknown to the players and $E$ is the number of players who entered.
Players who do not enter the market receive a stochastic payoff with zero mean.
After each round, each player is shown their reward, as well as the reward they
would have received by choosing the alternative.

Observations of human play were gathered by the CLER lab at
Harvard~\citep{erev2010}.  Each student involved in the experiment played ten
games lasting fifty rounds each.  The students were incentivized to play well
through a monetary reward proportional to their cumulative utility.  The parameter
$k$ was randomly selected in a fashion so that the Nash equilibrium had an
entry rate of $50\%$ in expectation.  In total, $30,000$ observations of play
were recorded.  The intent of the competition was to have teams submit programs
that would play in a similar fashion to the human subjects.  That is, the data
was used at test time to validate performance.  In contrast, our experiments
use actual observations of play at training time to build a predictive model of
the human behavior.  As we are interested in stationary behavior, we train and
test on only the last twenty five trials of each game.

We compared against two baselines.  The first baseline, labeled {\em
Multinomial} in the figures, is a smoothed multinomial distribution trained to
minimize the leave-one-out cross validation loss.  This baseline does not make
use of any features of the games.  That is, if the players indeed play
according to the Nash equilibrium we would expect this baseline to learn the
uniform distribution.  The second baseline, labeled {\em Logistic Regression}
in the figures, simply uses regularized logistic regression to learn a
linear classification boundary over the outcomes of the game using the same
features presented to our method.  Operationally, this is equivalent to using
MaxEnt Inverse Optimal Control in a single-agent setting where the utility is
summed across all the players.  This baseline has similar representational
power to our method, but lacks an understanding of the strategic elements of
the game.

\begin{figure}[ht] 
\centering
\includegraphics[width=.68\textwidth]{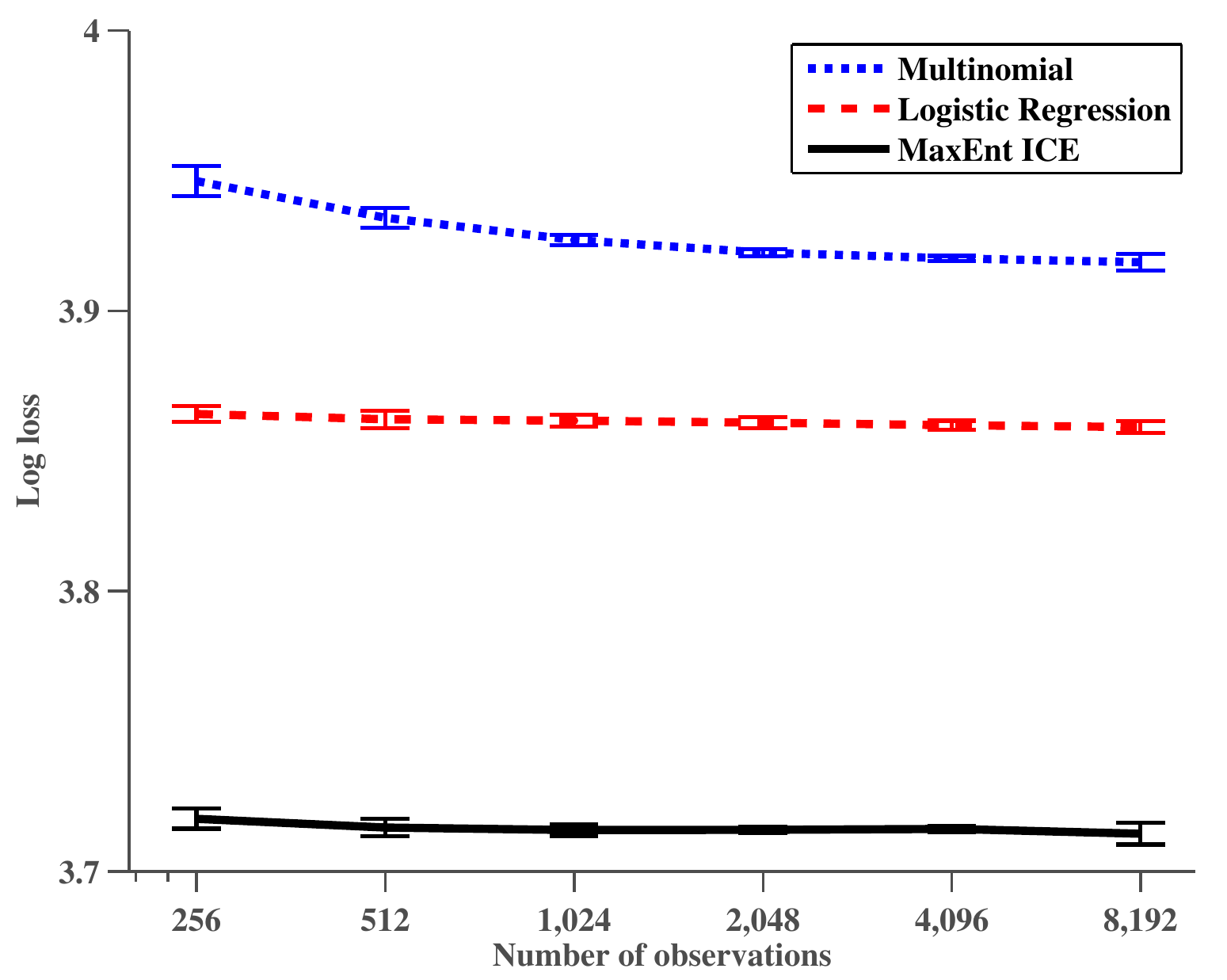} 
\caption{Test log loss using only the game's expected utility as a feature in
the market entry experiment.} 
\label{fig:utilitymle} 
\end{figure}

In Figure~\ref{fig:utilitymle}, we see a comparison of our method against the
baselines when only the game's true expected utility is used as the only
feature.  We see that our method outperforms both baselines across all sample
sizes.  We also observe the multinomial distribution performs slightly better
than the uniform distribution, which attains a log loss of $4$, though
substantially worse than logistic regression and our method, indicating that
the human players are not particularly well-modeled by the Nash equilibrium.
Our method substantially outperforms logistic regression, indicating that there
is indeed a strategic interaction that is not captured in the utility features
alone.

\begin{figure}[ht] 
\centering
\includegraphics[width=.68\textwidth]{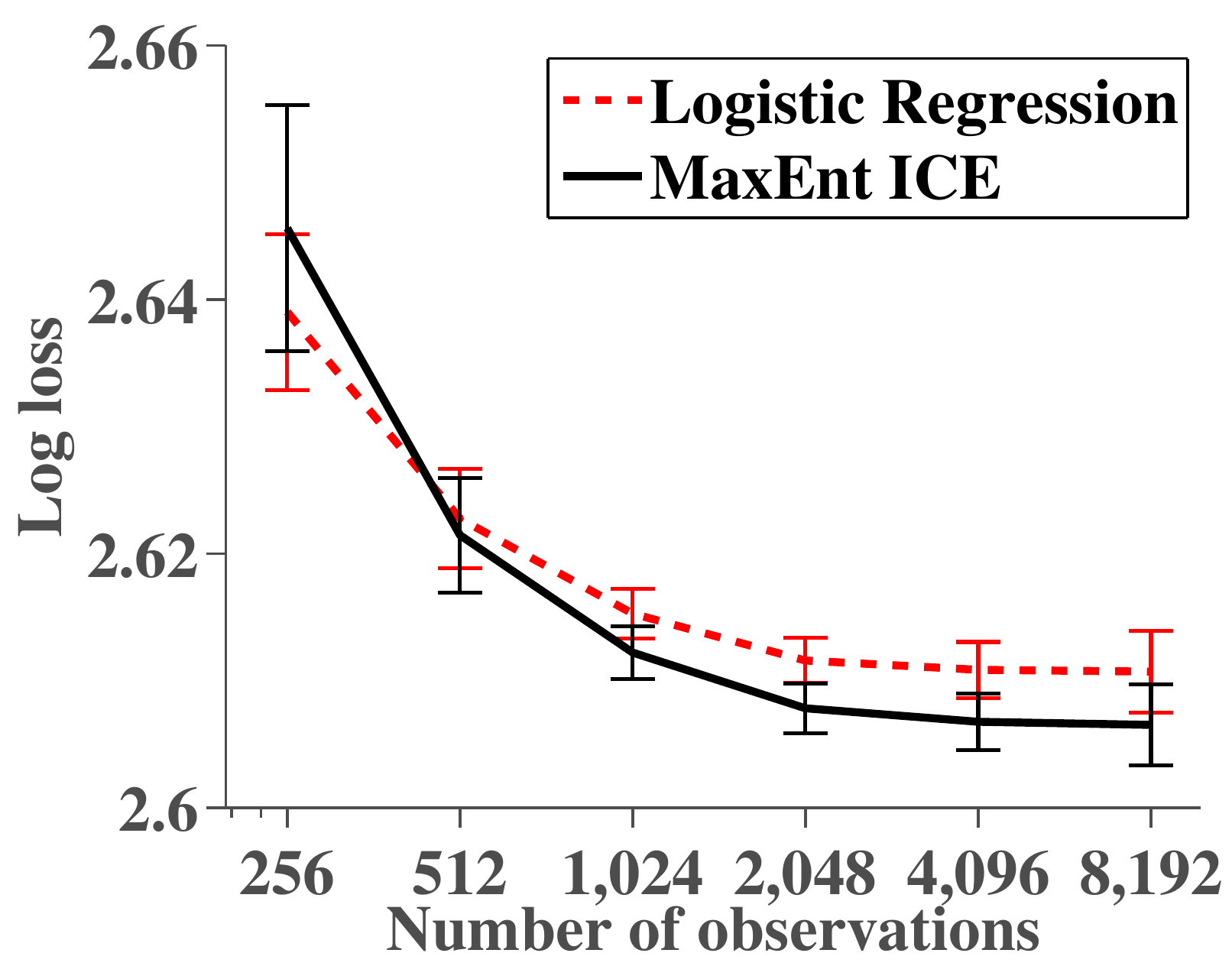} 
\caption{Test log loss using a number of history summary features in the market
entry experiment.} 
\label{fig:allfeatures} 
\end{figure}

In Figure~\ref{fig:allfeatures}, we see a comparison of our method against the
baselines using a variety of predictive features.  In particular, we summarize
a round using the observed action frequencies, average reward, and reward
variance up to that point in the round.  To weigh recent observations more
strongly, we also employ exponentially-weighted averages.  We observe that the
use of these features substantially improves the predictive power of the
feature-based methods.  Interestingly, we also note that the addition of these
summary features also narrows the gap between logistic regression and MaxEnt
ICE.  Under low sample sizes, the logistic model performs the best, but our
method overtakes it as more data is made available for training.  It appears
that in this scenario, much of the strategic behavior demonstrated by
the participants can be captured by these history features.

\subsection{Mid-scale Hotel Market Entry}

\begin{figure}[ht] 
\centering
\includegraphics[width=.6\textwidth]{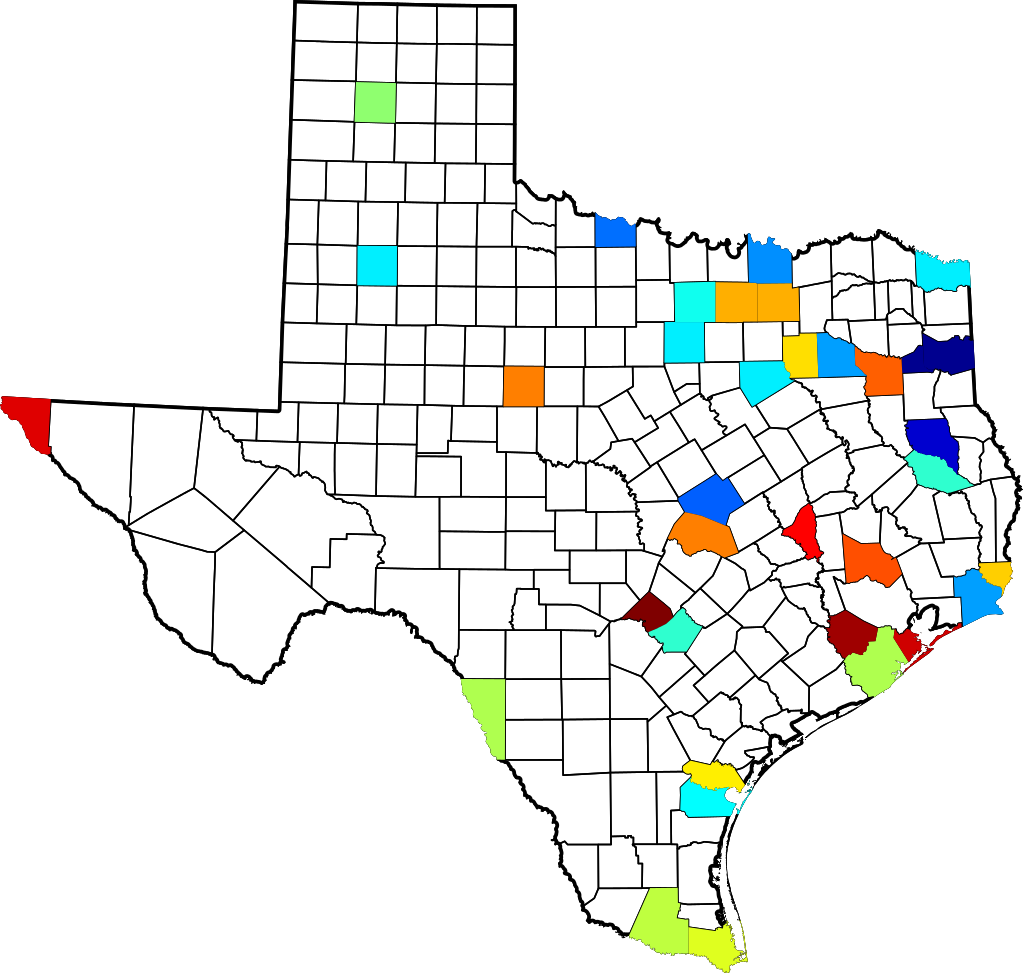} 
\caption{Regulatory index values for select counties in Texas.  Blue means little regulation
and lower costs to enter the market.  Red means higher costs.} 
\label{fig:texas}
\end{figure}

For our final experimental evaluation, we considered the task of predicting the
behavior of mid-scale hotel chains, like Holiday Inn and Ramada, in the state
of Texas.  Given demographic and regulatory features of a county, we wish to
predict if each chain is likely to open a new hotel or to close an existing
one.  The observations of play are derived from quarterly tax records over a
fifteen year period from forty counties, amounting to a total of $2,205$
observations.  The particular counties selected had records of all of the
demographic and regulatory features, had at least four action observations, and
none was a chain's flagship county.  Figure~\ref{fig:texas} highlights the
selected counties and visualizes their regulatory practices.

The demographic and regulatory features were aggregated from various sources
and generously provided to us by Prof. Junichi Suzuki~(\citeyear{suzuki2010}). The
demographic features for each county include quantities such as size of its
population and its area, employment and household income statistics, as well as
the presence or absence of an interstate, airport or national park.  The
regulatory features are indices measuring quantities such as commercial land
zoning restrictions, tax rates and building costs.  In addition to these noted
features, which are fixed across all time periods, there are time-varying
features such as the number of hotels and rooms for each chain and the
aggregate quarterly income.

We model each quarterly decision as a parameterized simultaneous-move game with
six players.  Each player, a mid-scale hotel chain, has the action set
$\{\text{Close},\text{NoAction},\text{Open}\}$, resulting in $729$ total
outcomes.  For the game's utility, we allocated the county's features to each
player in proportion to how many hotels they owned.  That is, if a player
operated 3 out of 10 hotels, the features associated with utility at that
outcome would be the county's feature vector scaled by $0.3$.  We included bias
features associated with each action to account for fixed costs associated with
opening or closing a hotel.

In the observation data, there are a small number of instances where a chain
opens or closes more than one hotel during a quarter.  These events are mapped
to $\text{Open}$ and $\text{Close}$ respectively.  Though the outcome-space is
quite large, the outcome distribution is extremely biased and the actions of
the chains are highly correlated.  In particular, over $80\%$ of time the time
no action is taken, around $17\%$ of the time a single chain acts, and less
than $3\%$ of the time more than one chain acts.  As a result, one expects the
featureless multinomial estimator to have reasonable performance despite a
large number of classes.

For experimentation, we evaluated four algorithms:  a smoothed multinomial
distribution trained to minimize the leave one out cross-validation loss,
MaxEnt inverse optimal control trained once for all players, multi-class
logistic regression over the joint action space, and regret-matching ICE with
utility matching constraints.  As the resulting optimizations for the latter two
algorithms are smooth, we employed the L-BFGS quasi-Newton method with
L2-regularization for training~\citep{nocedal80}.  As a substitute for
L1-regularization, we selected the $23$ best features based on their reduction
in training error when using logistic regression.  Each county had $63$
features available.  Of the top $23$ features selected, $11$ were regulatory
indices.

For the logistic regression and ICE predictors, we only used utility features
on the 13 high probability outcomes (no firms build, and one firm acting).  The
remaining outcomes had only bias features associated with them to help prevent
overfitting.  We experimented with a number of types of bias features, for
example, 4 bias features (one for no firms build, one for a single firm builds,
one for a single firm closes and one for all remaining outcomes), as well as
729 bias features (one for each outcome).  We found that, though on their own
the different bias features had varied predictive performance, when combined
with utility and regret features they were quite similar given the appropriate
regularization.  In the best performing model, which we present here, we used 729
bias features resulting in $1,028$ parameters to the logistic regression
model.

In the ICE predictor, we tied together the weights for each deviation across
all the players to reduce the number of model parameters.  For example, all
players shared the same dual parameters for the
$\text{NoAction}\rightarrow\text{Open}$ deviation.  Effectively, this alters
the rationality assumption such that the \emph{average} regret across all
players is the quantity of interest, instead of the maximum regret.
Operationally, this is implemented as summing each deviation's gradient in the
dual.  This treats the players anonymously, thus we implicitly and incorrectly
assume that conditioned on the county's parameters each firm is identical.
Due to the use parameter tying, the ICE predictor has an additional $156$
model parameters.

\begin{figure}[ht] 
\centering
\includegraphics[width=.6\textwidth]{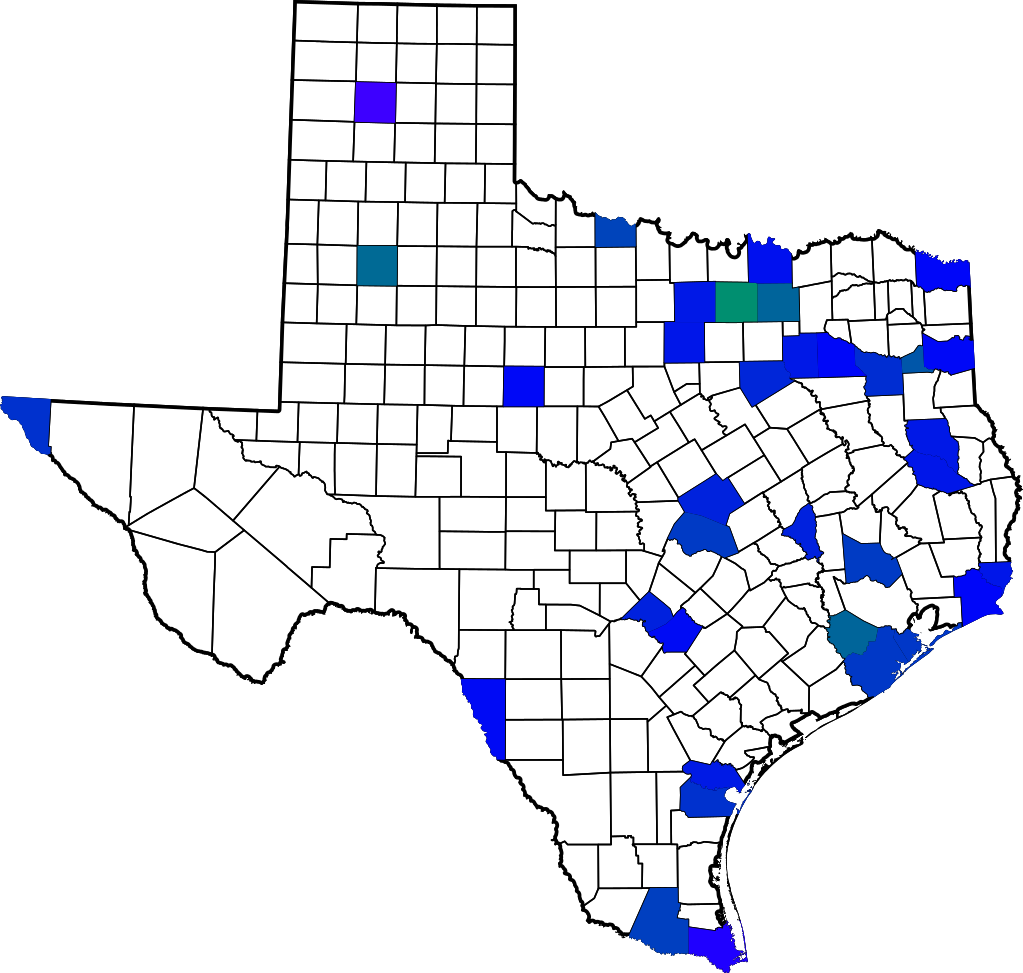} 
\caption{The marginalized probability that a chain will build a hotel in Spring 1996 predicted by MaxEnt ICE.  Brighter shades of green denote higher probabilities.} 
\label{fig:texas_predictions}
\end{figure}

The test losses reported were computed using ten-fold cross validation.  To fit
the regularization parameters for logistic regression, MaxEnt IOC and MaxEnt
ICE, we held out $10\%$ of the training data and performed a parameter sweep.
For logistic regression, a separate parameter sweep and regularization was
used for the bias and utility features.  For MaxEnt ICE, an additional
regularization parameter was selected for the regret parameters.  A sample
of the predictions from MaxEnt ICE are shown in Figure~\ref{fig:texas_predictions}.

\begin{figure} 
\centering 
\includegraphics[width=.32\textwidth]{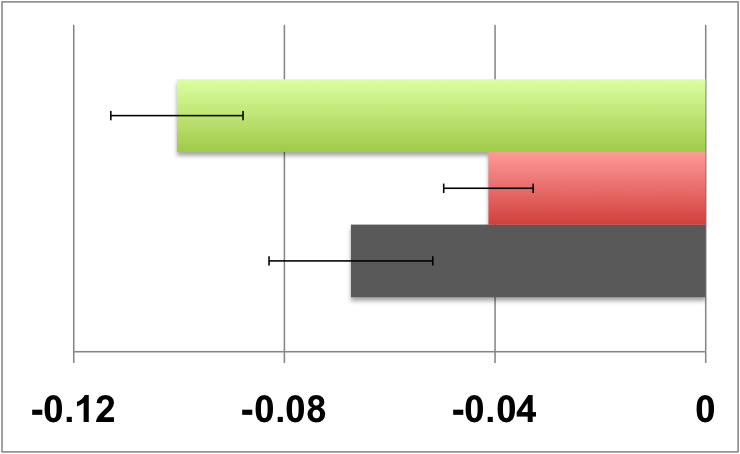}
\includegraphics[width=.32\textwidth]{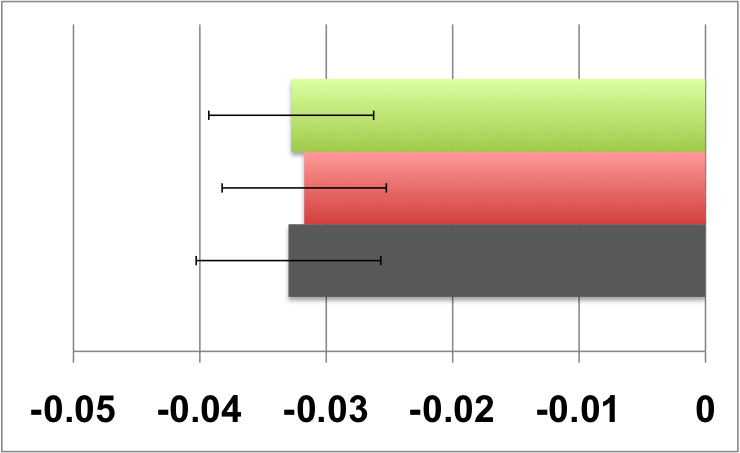}
\includegraphics[width=.32\textwidth]{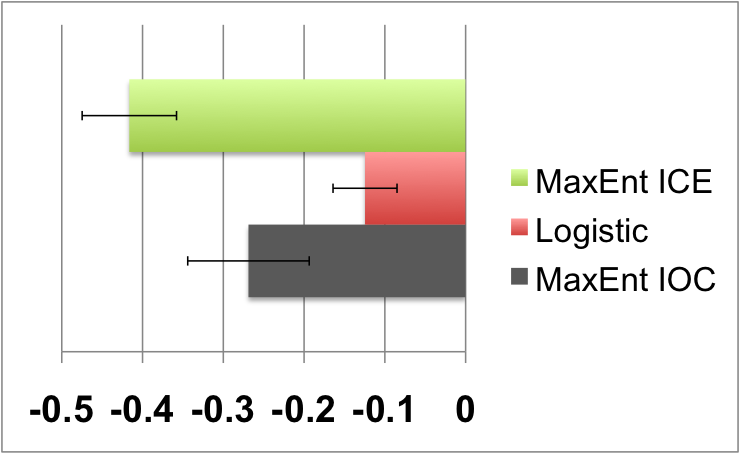} 
\caption{(Left) Test log loss on the full outcome space relative to
the smoothed multinomial, which has log loss $1.58234\pm 0.058088$.
(Center) Test log loss no build vs. build outcomes only.  Loss is
relative multinomial, with log loss $0.721466\pm 0.016539$. (Right)
Test log loss conditioned on build outcomes only.  Loss is relative
multinomial, with log loss $6.5911\pm 0.116231$.} 
\label{fig:hotel}
\end{figure}

In the left of Figure~\ref{fig:hotel}, we present the test errors of the three
parameterized methods in terms of their offset from that of the featureless
multinomial.  This quantity has lower variance than the absolute errors,
allowing for more accurate comparisons.  We see that the addition of the regret
features more than doubles the improvement of logistic regression 
from $2.6\%$ to $6.3\%$, where as the inverse optimal control method only sees
a $4.3\%$ improvement.

In the center of Figure~\ref{fig:hotel}, we show the test log-loss when the
methods are only required to predict if any firm acts.  Here, the models are
still trained over their complete outcome spaces and their predictions are
marginalized.  We see that all three methods are equal within noise.  That is,
the differences in the predictive performances come solely from each method's
ability to predict {\em who} acts.  We additionally performed this experiment
without the use of regulatory features and found that the logistic regression
method achieved a relative loss of $-0.027300$.  Using a paired comparison
between the two methods, we note that this difference of $0.004443$ is
significant with error $0.001886$.  This echoes Suzuki's conclusions the
regulatory environment in this industry affect firms' decisions to build new
hotels~\citep{suzuki2010}, measured here by improvements in predictive
performance.

In the right of Figure~\ref{fig:hotel}, we demonstrate the test log loss
conditioned on at least one firm acting---the portion of the loss that
differentiates the methods.  The logistic regression method with only utility
features performs the worst with a $1.8\%$ improvement over the multinomial
base line, the individual inverse optimal control method improves by $4.1\%$
and MaxEnt ICE performs the best with a $6.3\%$ improvement.  That is, the
addition of regret features, and hence accounting for the strategic aspects of
the game, have a significant effect on the predictive performance in this
setting.  We note that replacing the regulatory features in the regret portion
of the MaxEnt ICE model actually slightly improves performance to $-0.471763$,
though not by a significant margin.  This implies that the regulatory features
have little or no bearing on predicting exactly the firm that will act, which
suggests the regulatory practices are unbiased.

\section{Conclusion}

In this article, we develop a novel approach to behavior prediction in
strategic multi-agent domains. We demonstrate that by leveraging a rationality
assumption and the principle of maximum entropy our method can be efficiently
implemented while achieving good statistical performance.  Empirically, we
displayed the effectiveness of our approach on two market entry data sets.  We
demonstrated both the robustness of our approach to errors in our assumption as
well as the importance of considering strategic interactions.

Our future work will consider two new directions.  First, we will address
classes of games where the action sets and players differ. A key benefit of our
current approach is that it enables these to differ between training and
testing which we only leverage modestly in the transfer experiments for route
prediction.  This will involve investigating from a statistical point of view
novel notions of a deviation and their corresponding equilibrium concepts
Second, we will consider different models of interactions, such as stochastic
games and extensive-form games.  These models, though no more expressive than
matrix games, can often represent interactions exponentially more succinctly.
From a practical standpoint, this avenue of research will allow for the
application of our methods to a broader class of problems, including, for
instance, exploring the time series dependencies within the Texas Hotel Chain
data.

\section*{Acknowledgements}
This work is supported by the ONR MURI grant N00014-09-1-1052 and by the
National Sciences and Engineering Research Council of Canada (NSERC). The
authors gratefully acknowledge Prof. Junichi Suzuki for providing the
aggregated mid-scale hotel data and Alex Grubb for the application of density
estimation code to the data-sets.


\appendix
\section*{Appendix}

\subsection*{Proof of Lemma~\ref{lemma:internalapproxswap}}
\newcommand{\regretxysigmaw}{r^{\Gamesym}_{\playeri;\xsym\rightarrow\ysym}(\strategy|\wsym)}
\newcommand{\regretxxsigmaw}{r^{\Gamesym}_{\playeri;\xsym\rightarrow\xsym}(\strategy|\wsym)}
\begin{proof}
The lower bound holds as a consequence of $\devinternal\subseteq\devswap$.
\begin{align}
\mbox{Since~} \max_{\xsym\in\actionsi,\ysym\in\actionsi}\regretxysigmaw & \ge \regretxxsigmaw = 0, \\
\Regret{\Gamesym}{\devswap}{\strategy}{\wsym} & = \max_{\playeri\in\players}\sum_{\xsym\in\actionsi}\max_{\ysym\in\actionsi}\regretxysigmaw \\
& \le \numactions\cdot\max_{\playeri\in\players}\max_{\xsym\in\actionsi}\max_{\ysym\in\actionsi}\regretxysigmaw \\
& = \numactions\cdot\Regret{\Gamesym}{\devinternal}{\strategy}{\wsym}.
\end{align}
\end{proof}

\subsection*{Proof of Theorem~\ref{thm:kice}}
The proof of Theorem~\ref{thm:kice} immediately follows from the following lemma.
\begin{lemma}
For any utility function $\wsym\in\Ksym$, $\regretw{\Gamesym}{\deviation}{\predstrategy}{\wsym} \le \Regret{\Gamesym}{\devsetf}{\truestrategy}{\wsym}$ if and only if there exists an $\devdistf\in\simplex{\devsetf}$ such that $\regretf{\Gamesym}{\deviation}{\predstrategy} - \expectation{g\distributed\devdistf}{\regretf{\Gamesym}{g}{\truestrategy}} \in -\Kdualsym$.
\label{lemma:kice}
\end{lemma}
\begin{proof}
Assume that for all $\wsym\in\Ksym$, $\regretw{\Gamesym}{\deviation}{\predstrategy}{\wsym} \le \Regret{\Gamesym}{\devsetf}{\truestrategy}{\wsym}$, $\exists\devdistf\in\simplex{\devsetf}$ such that
\begin{align}
\regretw{\Gamesym}{\deviation}{\predstrategy}{\wsym} & \le \expectation{g\distributed\devdistf}{\regretw{\Gamesym}{g}{\truestrategy}{\wsym}}
\shortintertext{Since $0\in\Ksym$,}
\lefteqn{\max_{\wsym\in\Ksym} \regretw{\Gamesym}{\deviation}{\predstrategy}{\wsym} - \expectation{g\distributed\devdistf}{\regretw{\Gamesym}{g}{\truestrategy}{\wsym}} \le 0} \\
& = \regretw{\Gamesym}{\deviation}{\predstrategy}{0} - \expectation{g\distributed\devdistf}{\regretw{\Gamesym}{g}{\truestrategy}{0}} \\
& = \max_{\wsym\in\Ksym,t} \;\; \regretw{\Gamesym}{\deviation}{\predstrategy}{\wsym} - t,\;\mbox{subject to:~} t \ge \regretw{\Gamesym}{g}{\truestrategy}{\wsym}, \; \forall g\in\devsetf. \\
\shortintertext{By Slater's condition, strong duality holds and the resulting dual is the feasibility problem}
&= \min_{\devdistf\in\simplex{\devsetf}} \;\; 0, \;\mbox{subject to:~} \regretf{\Gamesym}{\deviation}{\predstrategy} - \expectation{g\distributed\devdistf}{\regretf{\Gamesym}{g}{\truestrategy}} \in -\Kdualsym. \\
\shortintertext{Assume $\exists \devdistf\in\simplex{\devsetf}$ and $y\in\Kdualsym$ such that $\regretf{\Gamesym}{\deviation}{\predstrategy}  - \expectation{g\distributed\devdistf}{\regretf{\Gamesym}{g}{\truestrategy}} + y = 0$, then for any $\wsym\in\Ksym$}
\regretw{\Gamesym}{\deviation}{\predstrategy}{\wsym} & - \expectation{g\distributed\devdistf}{\regretw{\Gamesym}{g}{\truestrategy}{\wsym}} + \inner{y}{\wsym} = \inner{0}{\wsym} = 0. \\
\shortintertext{By the definition of the dual cone $\inner{y}{\wsym}\ge 0$, therefore}
\regretw{\Gamesym}{\deviation}{\predstrategy}{\wsym} & \le \expectation{g\distributed\devdistf}{\regretw{\Gamesym}{g}{\truestrategy}{\wsym}} \le \max_{g\in\devsetf}{\regretw{\Gamesym}{g}{\truestrategy}{\wsym}} = \Regret{\Gamesym}{\devsetf}{\truestrategy}{\wsym}.
\end{align}
\end{proof}

\subsection*{Proof of Theorem~\ref{thm:nash}}
The proof of Theorem~\ref{thm:nash} immediately follows from the following lemma.
\begin{lemma}
If joint strategy $\strategy$ has $\eqmeps$ external regret, and $\Gamesym$ is 2-player and constant-sum with respect to $\wsym$, then the marginal strategies form a $2\eqmeps$-Nash equilibrium under utility function $\wsym$.
\end{lemma}

\newcommand{\utilityx}[1]{u^\Gamesym_{\xsym}(#1|\wsym)}
\newcommand{\utilityy}[1]{u^\Gamesym_{\ysym}(#1|\wsym)}
\newcommand{\strategyx}{\sigma^{\Gamesym}_{\xsym}}
\newcommand{\strategyy}{\sigma^{\Gamesym}_{\ysym}}
\newcommand{\strategymx}{\bar{\sigma}^{\Gamesym}_{\xsym}}
\newcommand{\strategymy}{\bar{\sigma}^{\Gamesym}_{\ysym}}
\newcommand{\actionsx}{A_{\xsym}}
\newcommand{\actionsy}{A_{\ysym}}
\begin{proof}
Denote one player $\xsym$ and the other $\ysym$ and their marginal strategies as $\strategymx$ and $\strategymy$ respectively.  We are given
\begin{align}
\forall \strategyx\in\simplex{\actionsx},\quad \utilityx{\strategyx, \strategymy} - \utilityx{\strategy} & \le \eqmeps \mbox{~and,} \\
\forall \strategyy\in\simplex{\actionsy},\quad \utilityy{\strategyy, \strategymx} - \utilityy{\strategy} & \le \eqmeps
\end{align}
as when either player deviates, the other resorts to playing his marginal strategy.  Substituting $\strategymy$ for $\strategyy$ and summing, we get $\forall \strategyx\in\simplex{\actionsx}$
\begin{align}
\utilityx{\strategyx, \strategymy} + \utilityy{\strategymy, \strategymx} - \left[\utilityy{\strategy} + \utilityx{\strategy}\right] & \le 2\eqmeps \\
\utilityx{\strategyx, \strategymy} + \utilityy{\strategymy, \strategymx} - C & \le 2\eqmeps \\
\utilityx{\strategyx, \strategymy} + \left[C - \utilityx{\strategymy, \strategymx}\right] - C & \le 2\eqmeps \\
\utilityx{\strategyx, \strategymy} - \utilityx{\strategymy, \strategymx} & \le 2\eqmeps
\end{align}
A symmetric argument shows the equivalent statement for the opposing player.
\end{proof}

\subsection*{Proof of Theorem~\ref{thm:dual}}
\begin{proof}
The Legrange dual function,
\begin{align}
L(\dualweights,\alpha,\beta,u,v,x) = \;\;\max_{\predstrategy,\devdist,y} & -\sum_{\outcome\in\outcomes}\predstrategya{\outcome}\log\predstrategya{\outcome} + \\
& \sum_{\deviation\in\devset}\alpha_f\left(1- \sum_{g\in\devset}\eta_f(g) \right) + \beta\left(1- \sum_{\outcome\in\outcomes}\predstrategya{\outcome}\right) + \\
& \sum_{\deviation\in\devset}\inner{\regretf{\Gamesym}{\deviation}{\predstrategy} - \sum_{\deviation,g\in\devset}\eta_f(g)\regretf{\Gamesym}{g}{\truestrategy} + y_f}{\dualweightf} + \\
& \sum_{\deviation\in\devset}v_f\cdot\eta_f + u\cdot\predstrategy + \inner{y_f}{x_f},
\end{align}
is an upper bound on the primal objective for all $u,v\ge 0$ and $x\in\Kddualsym$.  We solve the unconstrained maximization by setting the derivatives with respect to $\predstrategy,\devdist,y$ to zero,
\begin{align}
\log\predstrategya{\outcome} & = u(a) - 1 - \beta - \sum_{\deviation\in\devset}\inner{\regretf{\Gamesym}{\deviation}{\outcome}}{\dualweightf} \\
\inner{\regretf{\Gamesym}{g}{\truestrategy}}{\dualweightf} -\alpha_f + v_f(g) & = 0 & \forall \deviation,g\in\devset \\
\dualweightf & = -x_f. & \forall \deviation\in\devset
\end{align}
Substituting this solution back into the Legrangian and minimizing this upper bound gives
\begin{align}
\min_{\dualweights\in\Kddualsym,\beta} & \;\;\exp(-1 - \beta)\Zfunc{\Gamesym} + \sum_{\deviation\in\devset}\alpha_f + \beta \quad\mbox{subject to:} \\
& \inner{\regretf{\Gamesym}{g}{\truestrategy}}{\dualweightf}  \le \alpha^f. & \forall \deviation,g\in\devset
\end{align}
Solving for $\beta$ explicitly we get $\beta = \log \Zfunc{\Gamesym} - 1$, and moving the constraint into the objective gives our result:
\begin{align}
\min_{\dualweights\in\Kddualsym} & \;\; \sum_{\deviation\in\devset}\Regret{\Gamesym}{\devset}{\truestrategy}{\dualweightf} + \log\Zfunc{\Gamesym}.
\end{align}
\end{proof}

\subsection*{Proof of Theorem~\ref{thm:samplefinite}}
\begin{proof}
Let $\{e_1,e_2,\ldots,e_K\}$ be an orthonormal basis for $\Vsym$, where $K=\dim{\Vsym}$.  We first bound how well the regrets match in each basis direction.
\begin{align}
P\Big[ \max_{\deviation\in\devset,k\in [K]} \abs{\regretw{\demongamedist}{\deviation}{\demonstrategies}{e_k} - \regretw{\truegamedist}{\deviation}{\truestrategies}{e_k}} \ge \epsilon\maxregret  \Big] & \le P\Big[ \bigcup_{\deviation\in\devset, k\in [K]}\big( \abs{\regretw{\demongamedist}{\deviation}{\demonstrategies}{e_k} - \regretw{\truegamedist}{\deviation}{\truestrategies}{e_k}} \ge \epsilon\maxregret \big) \Big] \\
& \le \sum_{\deviation\in\devset, k\in [K]} P\Big[ \abs{\regretw{\demongamedist}{\deviation}{\demonstrategies}{e_k} - \regretw{\truegamedist}{\deviation}{\truestrategies}{e_k}} \ge \epsilon\maxregret \Big] \\
& \le 2\abs{\devset}K\exp\left(-2\Tsym\epsilon^2\right) \le \delta. \\
\Tsym & \ge \frac{1}{2\epsilon^2}\log\frac{2\abs{\devset}\dim{\Vsym}}{\delta}.
\intertext{Next, we bound how well the regrets match under $\wsym$, given all regrets are close.}
\abs{\regretw{\demongamedist}{\deviation}{\demonstrategies}{\wsym} - \regretw{\truegamedist}{\deviation}{\truestrategies}{\wsym}} & = \abs{\sum_{k=1}^K \alpha_k\regretw{\demongamedist}{\deviation}{\demonstrategies}{e_k} - \alpha_k\regretw{\truegamedist}{\deviation}{\truestrategies}{e_k}} \\
& \le \sum_{k=1}^K \abs{\alpha_k}\abs{\regretw{\demongamedist}{\deviation}{\demonstrategies}{e_k} - \regretw{\truegamedist}{\deviation}{\truestrategies}{e_k}} \\
& \le \epsilon\maxregret\sum_{k=1}^K \abs{\alpha_k} \\
& \le \epsilon\maxregret\norm{\wsym}_1.
\end{align}
\end{proof}

\subsection*{Proof of Theorem~\ref{thm:sampleinf}}
\begin{proof}
Unlike Theorem~\ref{thm:samplefinite}, we will bound the error of regrets directly.
\begin{align}
P\Big[ \max_{\deviation\in\devset} \regretw{\demongamedist}{\deviation}{\demonstrategies}{\wsym} - \regretw{\truegamedist}{\deviation}{\truestrategies}{\wsym} \ge \epsilon\maxregretf{\wsym}  \Big] & \le P\Big[ \bigcup_{\deviation\in\devset}\big( \regretw{\demongamedist}{\deviation}{\demonstrategies}{\wsym} - \regretw{\truegamedist}{\deviation}{\truestrategies}{\wsym} \ge \epsilon\maxregretf{\wsym}  \big) \Big] \\
& \le \sum_{\deviation\in\devset} P\Big[ \regretw{\demongamedist}{\deviation}{\demonstrategies}{\wsym} - \regretw{\truegamedist}{\deviation}{\truestrategies}{\wsym} \ge \epsilon\maxregretf{\wsym}  \Big] \\
& \le 2\abs{\devset}\exp\left(-2\Tsym\epsilon^2\right) \le \delta \\
\Tsym & \ge \frac{1}{2\epsilon^2}\log\frac{\abs{\devset}}{\delta}.
\end{align}
\end{proof}

\vskip 0.2in
\bibliographystyle{abbrvnat}
\bibliography{paper}

\end{document}